\documentclass[11pt]{article}
\usepackage[margin=1in]{geometry}
\usepackage{graphicx} 
\usepackage[dvipsnames]{xcolor}

\usepackage{amsmath,amssymb,amsthm}
\usepackage{mathtools}

\usepackage{algorithm}
\usepackage{algorithmicx}
\usepackage{algpseudocode}
\usepackage{braket}
\usepackage{authblk}

\usepackage{hyperref}

\usepackage[sorting=none]{biblatex}
\title{How fast can a parent estimate the value of their children?\\ A quantum algorithm for stochastic games}

\author[1,2]{Marien Raat\thanks{\protect\hypertarget{equalcontrib}{Marien Raat and Merlin Incerti-Medici contributed equally to this paper}}}
\author[3]{Merlin Incerti-Medici\textsuperscript{*}}
\author[3]{James R. Wootton}
\author[1,2,4]{Evert van Nieuwenburg}
\author[3]{Daniel Bultrini}

\affil[1]{$\langle aQa ^L\rangle $ Applied Quantum Algorithms, Universiteit Leiden}
\affil[2]{Instituut-Lorentz, Universiteit Leiden, Niels Bohrweg 2, 2333 CA Leiden, Netherlands}
\affil[3]{Moth, Switzerland}
\affil[4]{LIACS, Universiteit Leiden, Niels Bohrweg 1, 2333 CA Leiden, Netherlands}

\date{September 2026}

\newtheorem{definition}{Definition}
\newtheorem{remark}{Remark}
\newtheorem{theorem}{Theorem}
\newtheorem{example}{Example}
\newtheorem{lemma}{Lemma}
\newtheorem{assumption}{Assumption}
\newtheorem{corollary}{Corollary}
\newtheorem{proposition}{Proposition}

\newif\ifauthorcolours
\authorcoloursfalse
\newcommand{\vio}[1]{\ifauthorcolours{\color{violet}#1}\else{#1}\fi}
\newcommand{\grn}[1]{\ifauthorcolours{\color{ForestGreen}#1}\else{#1}\fi}
\newcommand{\mrn}[1]{\ifauthorcolours{\color{Maroon}#1}\else{#1}\fi}
\newcommand{\authorcolour}[1]{\ifauthorcolours\color{#1}\fi}

\begin{document}

\maketitle

\begin{abstract}
\grn{We give a quantum algorithm for stochastic $2$-player games. Their game trees are the expectiminimax trees of classical game search: $m$ adversarial layers alternate with $m$ chance layers, every node has $deg$ children, and the leaf values lie in an interval of length $N$. The algorithm estimates the value of the game to root mean square error $\epsilon$ with
\[ O\!\left( K^{D}\, deg^{\frac{D}{4}} \, \frac{N}{\epsilon} \, D^{4D} \log\left( \frac{deg N}{\epsilon} \right)^{4D} \right), \qquad D = 2m \]
queries to the leaf values, where $K$ is an absolute constant. If we treat the depth $D$ as a constant as is sometimes done in the literature, the bound is $\widetilde O(deg^{m/2}\epsilon^{-1})$ against the classical $\widetilde O(deg^{m}\epsilon^{-2})$. So the quadratic speedup in the branching factor known for adversarial trees and the quadratic speedup in the accuracy known for a single expectation both survive when we nest one inside the other. 
We use a derandomised multilevel Monte Carlo estimator for chance layers and a coherent binary search for adversarial layers. The composition is done by introducing a conversion between root-mean-square and uniform error guarantees, and a composition to nest it with quantum mean estimation. We state the algorithm for stochastic games, but it is valid for the following conditions: the value of a parent is a Lipschitz function of the values of its children and the value at a chance vertex is linear in its children, however the speedup is preserved only if the deterministic step has a quantum subroutine with sublinear query complexity.}
\end{abstract}

\section{Introduction} \label{sec:introduction}

\grn{Stochastic games model two opposing players who play turn by turn with a random process which intervenes between their moves: such as a die being thrown in Backgammon or EinStein w\"urfelt nicht. We want to compute the payoff both players can guarantee under optimal play, known as the value of the game. This can be defined by a recursion where you alternate $\max$, $\min$ and an expectation over the chance outcomes down the decision tree. Such a tree is what is solved \grn{classically} by \mrn{the expectiminimax algorithm~\cite{michieCHAPTER8GAMEPLAYING1966}}. At a $\max$ node player one, at a $\min$ node player two moves, and at a chance node a stochastic process changes the game state. Evaluating the root exactly means reading every leaf \mrn{in the worst case}, since \mrn{a single} leaf can change the answer. We ask how many leaf evaluations a quantum algorithm needs to estimate the value \mrn{to an additive error $\epsilon$ with a success probability of at least $\frac{2}{3}$ in the worst case}.}

\grn{Classically, even with pruning, each deterministic layer costs a constant power of the branching factor $deg$, and each chance layer costs $\epsilon^{-2}$ samples to reach accuracy $\epsilon$. Two quantum speedups exist for both types of problems separately. At a deterministic node, an extremum over $deg$ children takes $O(\sqrt{deg})$ queries instead of $O(deg)$; for pure minimax \grn{trees} this is known and optimal~\cite{cleveQuantumAlgorithmsEvaluating2019,barnumLowerBoundQuantum2004}. At a chance node, an expectation takes $O(\epsilon^{-1})$ samples instead of $O(\epsilon^{-2})$~\cite{kothariMeanEstimationSourceCode2023}. The problem is that we must nest \mrn{these two steps} inside one another $D$ times, and a naive composition of the two algorithms loses both speedups. If we divide an error budget $\epsilon$ equally over $c$ chance layers and run amplitude estimation at each the cost is $O((c/\epsilon)^{c})$, which is worse than sampling. If we simply nest extremum finding, each child must be far more accurate than its parent. For pure nested expectations, a derandomised multilevel Monte Carlo scheme solves the first problem: it telescopes the level differences and brings the total cost of all chance layers back to $O(\epsilon^{-1})$ \mrn{with logarithmic overhead}~\cite{sunOptimalQuantumSpeedups2026,blanchetNonlinearQuantumMonteCarlo2025}. This paper extends that scheme to trees whose layers are not all expectations and shows that we can search the extremum layers instead of evaluating them along the way.}

\grn{We show that through our scheme both speedups survive. For the chance layers we use the derandomised multilevel estimator of \cite{sunOptimalQuantumSpeedups2026}. The value at a chance node is a telescoping sum of level differences. The second moment of the level-$n$ difference falls as $2^{-n}$ and the cost of one sample rises as $2^{n/2}$, so the number of samples times their cost does not depend on the level, and the whole stack of chance layers costs $\epsilon^{-1}$ instead of $\epsilon^{-c}$. For the deterministic layers we use a coherent binary search over the value. It tests each pivot by amplitude amplification over the children with a fixed schedule, finds the extremum of $deg$ children with $O(\sqrt{deg})$ queries, and, unlike D\"urr--H\o yer~\cite{durrQuantumAlgorithmFinding1999}, needs no intermediate measurement. To make the two compose we need two components. First, we must convert between the root-mean-square guarantee given from a mean estimator and the uniform guarantee required by an extremum (Lemma~\ref{lem:median.bridge}). Second, an extremum over $deg$ children must not \mrn{amplify the root-mean-square error by} a factor $\sqrt{deg}$. Under only a mean-square hypothesis this loss cannot be avoided (Remark~\ref{rem:max.lemma.is.tight}); the binary search avoids it by bracketing the \emph{true} extremum instead of taking an extremum of estimates (Lemma \ref{lem:search.step}). \mrn{If $m$ denotes the number of deterministic layers and $T$ is a tree} of depth $D = 2m$ with branching factor $deg$ and leaf values in an interval of length $N$, the result (Theorem \ref{thm:headline}) is}
\[ \grn{Q_{\epsilon}(T) \; \leq \; O\!\left( K^{D}\, deg^{\frac{D}{4}} \, \frac{N}{\epsilon} \, D^{4D} \log\left( \frac{deg N}{\epsilon} \right)^{4D} \right).} \]
\grn{queries to the leaf oracle, where $K$ is an absolute constant. The bound can be read in two ways. If the depth is a constant, as in \cite{sunOptimalQuantumSpeedups2026} and in several nested expectation results, \grn{it} can be thought of as $\widetilde O(deg^{m/2}\epsilon^{-1})$ against the classical $\widetilde O(deg^{m}\epsilon^{-2})$ of (\ref{eq:classical.baseline}): a quadratic improvement in both parameters (Corollary \ref{cor:constant.depth}). If we keep the depth dependence, the improvement in $\epsilon$, from $\epsilon^{-2}$ to $\epsilon^{-1}$, holds for every tree, and the improvement in $deg$ holds for trees that are wide compared to their depth, $deg \gtrsim K^4\Lambda^{12}$ (Remark \ref{rem:regime}). We prove the bound for any linear passing function at the chance layers and any  \vio{Lipschitz passing function with a} sublinear extremum routine at the deterministic layers (Theorem \ref{thm:unitary.general}), so the algorithm also applies to stochastic processes that are not games. On the lower-bound side, \mrn{the cost of extremum finding and of mean estimation are} necessary on \mrn{their} own, and a composed hard instance shows that for $\epsilon \geq \frac{1}{4\,deg}$ their product is \mrn{unavoidable} up to a factor $\sqrt{deg}$ (Proposition \ref{prop:composed.lower.bound}).}

\mrn{To our knowledge,} our work contributes the following results:
 \grn{we state the abstract optimization problem on a layered tree whose passing functions $G_d$ alternate between a deterministic and non-deterministic nodes (expectiminimax) directly as a recursive family of unitaries (Algorithm \ref{alg:generalized.unitary.estimator.constructor}). We then bound its query complexity and root-mean-square error in terms of two parameters of the deterministic step and one of the non-deterministic step (Theorem \ref{thm:unitary.general}).}

 \grn{We introduce two interface lemmas that allow the layers to compose. Lemma \ref{lem:median.bridge} converts between the mean-square guarantee given by a chance layer and the uniform guarantee required by deterministic layers. Remark \ref{rem:max.lemma.is.tight} shows that no argument that uses mean-square guarantees alone can avoid losing $\sqrt{deg}$ at an extremum.}

\grn{These are exploited in the coherent extremum routine (Algorithm \ref{alg:min-max-binary-search}) with a correctness proof (Lemma \ref{lem:search.step}) that needs no promise on where any child lies relative to the pivot. This is necessary because another initially promising route through bounded-error quantum search \cite{hoyerQuantumSearchBoundedError2003} fails for that reason.}

\grn{Although one may think of many approaches, we present a composition lemma (Lemma \ref{lem:coherent.composition}) that justifies nesting the layers inside quantum mean estimation as an approach. \grn{Because} the circuit and all its parameters do not depend on the vertex, it proves why variable-time subroutines such as D\"urr--H\o yer cannot be used.}

\grn{Finally we show a regime in which the bound is an improvement (Remark \ref{rem:regime}). We show lower bounds in each parameter on their own; a lower bound on an instance in which both kinds of layer are active (Proposition \ref{prop:composed.lower.bound}), which matches the upper bound up to one factor $\sqrt{deg}$ whenever $\epsilon \geq \frac{1}{4\,deg}$; and an unconditional fallback (Corollary \ref{cor:unconditional}) that separates the speedup in $\epsilon$ from the speedup in $deg$.}

\grn{It goes without saying that we build on a wealth of research in this area. Quantum evaluation of AND-OR and minimax trees is well understood. \textcite{farhiQuantumAlgorithmHamiltonian2007,childsDiscretequeryQuantumAlgorithm2009,ambainisAnyAndOrFormula2010,reichardtFasterQuantumAlgorithm2011} give $\grn{n}^{1/2+o(1)}$ algorithms for formulas of size $\grn{n}$, \textcite{cleveQuantumAlgorithmsEvaluating2019} extends this to Min-Max trees over an ordered set, and \textcite{barnumLowerBoundQuantum2004} gives the matching $\Omega(\sqrt \grn{n})$ lower bound. Classically, the best randomised exponent for a balanced binary tree is $\grn{n}^{0.753}$ \cite{saksProbabilisticBooleanDecision1986}, and for a $k$-ary tree the saving over reading every leaf is a factor of about $2$ per layer (Section \ref{sec:complexity.measure}). However, these don't cover chance nodes. In contrast, quantum mean estimation \cite{brassardQuantumAmplitudeAmplification2002,montanaroQuantumSpeedupMonteCarlo2015,kothariMeanEstimationSourceCode2023,hamoudiQuantumAlgorithmsMonte} and its multilevel and nested extensions \cite{blanchetNonlinearQuantumMonteCarlo2025,sunOptimalQuantumSpeedups2026} handle nested expectations but not extrema. The Lipschitz hypothesis of \cite{sunOptimalQuantumSpeedups2026} concerns one scalar continuation value; our setting needs a vector of $deg$ children. Tree-size estimation \cite{ambainisQuantumAlgorithmTree2017} and quantum backtracking \cite{montanaroQuantumWalkSpeedup2016,seidelQuantumBacktrackingQrisp2024,rennelaHybridDivideandconquerApproach2023a} attack game trees from a different direction. The classical literature on chance nodes in game trees begins with \cite{ballardStarMinimaxSearch1983}.}

\grn{The paper is organized as \grn{follows}: Section~\ref{sec:setting} defines the problem, the assumptions, the complexity measure and the classical baseline. Section~\ref{sec:prelim} collects the moment inequalities and the interface lemmas. Section~\ref{sec:generalized.algorithm} gives the abstract algorithm as a recursive family of unitaries and proves its guarantee. Section~\ref{sec:unitary} gives the two layer algorithms for stochastic $2$-player games, the extremum search and the mean estimator, shows that they satisfy the abstract assumptions, and derives the main bound. Section~\ref{sec:optimality} gives the lower bounds and the regime of advantage. Section~\ref{sec:ideas} lists open problems.}

\section{Setting} \label{sec:setting}

\subsection{Trees, passing functions and values}

Consider the following optimization problem. We have a rooted tree \mrn{$T$} with $D+1$ layers (the layer of the root is layer $0$). For simplicity, we assume that every vertex at layer $d$ has exactly $deg_d$ many children. We have a map $type: \{ 0, \dots, D-1 \} \rightarrow \{ det, \grn{non\text{-}det} \}$ such that $type(d) \neq type(d+1)$ for every $d$. The value of $type(d)$ indicates whether the step from layer $d$ to layer $d+1$ is deterministic or non-deterministic. \grn{We write $V(T)_d$ for the set of vertices at layer $d$.}

For every $d \in \{ 1, \dots, D \}$, we have maps
\[ G_d: V(T)_{d-1} \times \mathbb{R}^{deg_d} \rightarrow \mathbb{R}, \]
\[ (v, x_1, \dots, x_{deg_d}) \mapsto G_d(v, x_1, \dots, x_{\grn{deg_d}}). \]
We think of $G_d(v,\bar{x})$ as the `value' at the vertex $v$ assuming that its children $w_1, \dots, w_{deg_d}$ have values $x_1, \dots, x_{deg_d}$ respectively. (We assume that we have once and for all fixed an ordering of the children of every vertex.) We call the $G_d$ the {\it passing functions}.

For every $d \in \{ 1, \dots, D+1 \}$, we define
\[ \gamma_d: V(T)_{d-1} \rightarrow \mathbb{R} \]
via a recursive formula. Namely, for $d \leq D$ and $v \in V(T)_{d-1}$, we define
\begin{equation} \label{eq:general.recursive.formula}
    \gamma_d(v) = G_d(v, \gamma_{d+1}(w_1), \dots, \gamma_{d+1}(w_{deg_d})).
\end{equation}
For $d = D+1$, we introduce the notation $val = \gamma_{d}$. The function $val$ is given to us together with the tree $T$ and the functions $G_d$. The purpose of our algorithm is to compute $\gamma_d(v)$ for all the non-leaves of the tree with high accuracy. \grn{The output is the estimate at the root; the recursion produces a unitary for every layer.}

We write $R_d(v,\epsilon)$ for an estimator of $\gamma_d(v)$ with root-mean-square error (RMSE) $\leq \epsilon$, meaning that \[\mathbb{E}\left[ \left(R_d(v,\epsilon) - \gamma_d(v) \right)^2 \right] \leq \epsilon^2.\]

\vio{The index $d$ of $\gamma_d$, $R_d$ and $G_d$, and later of the unitaries $U_d$, is the index of the \emph{step} from layer $d-1$ to layer $d$, not the index of a layer. All four concern a vertex $v \in V(T)_{d-1}$ and its children in $V(T)_d$. So $R_d(v,\epsilon)$ estimates $\gamma_d(v) = G_d(v, \gamma_{d+1}(w_1), \dots, \gamma_{d+1}(w_{deg}))$ for $v$ at layer $d-1$, and it is built from the child estimators $R_{d+1}(w_i,\cdot)$ at layer $d$. This is also why $type(d-1)$, the type of the step into layer $d$, decides how $R_d$ is built.}

Recall that we assume $type$ to be alternating. \mrn{Additionally we assume} that $type(D-1) = det$, that is, the last step is deterministic. To simplify the formulation of our algorithm, we assume that $deg_d = deg$ is independent of $d$. \vio{We only need} a uniform upper bound $deg_d \leq deg$ for the complexity analysis. \vio{This is because every cost factor in this work is non-decreasing in degree, so replacing a layer-dependent degree by the maximum can only loosen the bound.}

We introduce the following notation to make the indication of children easier. We write
\[ G_d(v, w ):= G_d(v, w_1, \dots, w_{deg_d}),\]
where $w_1, \dots, w_{deg_d}$ are the children of $v$. Similarly, we write
\[ G_d(v, f(w) ):= G_d(v, f(w_1), \dots, f(w_{deg_d}) ) \]
if $f$ is a function defined on all the children of $v$. \mrn{For example, we have $G_d(v, R_{d+1}(w,\epsilon))$.}

\subsection{Expectiminimax trees}

\begin{definition}[Expectiminimax tree] \label{def:emm.tree}
    A $(k, d)$-expectiminimax tree (EMM tree) is a complete k-ary tree with depth $d$ and \grn{$k^d$} leaves, along with a function $\tau: \{0,..., d-1\} \to \{MAX, MIN, CHANCE\}$ where $\tau(i)$ is the node type of every node at depth $i$ from the root. \grn{Indexing $\tau$ by depth from the root rather than from the leaves is a convention. We keep it because the complexity statements refer to the depth of the vertex at hand.} \grn{In the language of the passing functions above, $\tau(i) = MAX$ or $MIN$ means $type(i) = det$ with $G_{i+1}$ the maximum or the minimum, and $\tau(i) = CHANCE$ means $type(i) = non\text{-}det$ with $G_{i+1}$ the mean.}
\end{definition}

\begin{definition}[Value of an EMM tree]
When numerical values $x \in \grn{[0,1]^{\grn{k^d}}}$ are assigned to the \grn{leaves} of an EMM tree, the value of the tree is defined recursively as follows:

\begin{equation}
    val(v, x) = \begin{cases}
        x_v & \text{if $v$ is a leaf} \\
        max_{j=1}^k val(v_j, x) & \text{if $\tau(depth(v)) = $ Max} \\
        min_{j=1}^k val(v_j, x) & \text{if $\tau(depth(v)) = $ Min} \\
        \frac{1}{k} \sum_{j=1}^k val(v_j, x) & \text{if $\tau(depth(v)) = $ Chance}
    \end{cases}
\end{equation}
\grn{Here $v_1, \dots, v_k$ are the children of $v$ in a fixed ordering, $depth(v)$ is the number of edges from the root to $v$, and $x_v$ is the value assigned to the leaf $v$.}
\end{definition}

\mrn{The function $val(v,x)$ corresponds to the function $\gamma_d(v)$, where $v \in V(T)_{d-1}$ and $x$ determines $\gamma_{D+1}$. This allows us to see expectiminimax trees as an example of the problem of determining the $\gamma_d(v)$ for a concrete tree and value-passing functions $G_d \in \{ \min, \max, \mathrm{mean} \}$.}

\begin{remark}
    We use uniform weights for the children of chance nodes. \grn{Nothing depends on this. The general passing functions of Section \ref{sec:generalized.algorithm} cover any weight vector $p$, for which the Lipschitz constant of the chance node is $\|p\|_\infty$ in the $\ell^1$-norm and $1$ in the sup-norm, and Remark \ref{lem:non.det.assumption.is.satisfied} draws the child from $p$ instead of averaging uniformly.}
\end{remark}

\subsection{No linear blowup and normalization}

Suppose we assign values $\gamma_{\grn{D+1}}(v)$ to every leaf $v$ of the rooted tree. In general, the repeated application of the functions $G_d$ can assign values of arbitrary size to the other vertices in the tree.

\begin{definition} \label{def:no.blowup}
    We say that a collection $(T, (G_d)_d)$ {\it has no \grn{linear} blowup}, if there exists a constant $C$ such that whenever the values $(\gamma_{\grn{D+1}}(v))_{v \text{ \mrn{a} leaf}}$ lie in an interval of length $N$, then the values \grn{$\gamma_d(v)$ for all $d \in \{1,\dots,D+1\}$ and all $v \in V(T)_{d-1}$} lie in an interval of length $CN$.
\end{definition}

\mrn{An example that has no linear blowup are expectiminimax trees: It is straight-forward to see that the value-passing functions in these trees keep all values within the interval of the initial values at the $D$-th layer.} Whenever a tree and its passing functions have no \grn{linear} blowup and we have a selection of values $\gamma_{\grn{D+1}}(v)$ for the leaves, we can \grn{shift and} rescale these values to guarantee that all values $\gamma_{\grn{d}}(v)$ for all vertices in $T$ lie in \grn{$[0,1]$}. This renormalization is often useful, although our results apply whenever there is no \grn{linear} blowup. \grn{We use this normalization throughout the unitary formulation (Sections \ref{sec:generalized.algorithm} and \ref{sec:unitary}). An RMSE of $\epsilon$ in the original units is an RMSE of $\epsilon/\grn{(CN)}$ after rescaling\grn{, with $C$ the constant of Definition \ref{def:no.blowup}; $C = 1$ for $\max$, $\min$ and the mean, the only case we need}, which is where the factor $N$ in the final statements comes from.}

\subsection{Lipschitz continuity}

\begin{assumption} \label{ass:generalized.Lipschitz}
    We assume that there exist constants $L_{all} \geq L > 0$ such that for every $d \in \{ 1, \dots, D \}$, $G_d(v, x_1, \dots, x_{deg_d})$ is $L_{all}$-Lipschitz continuous in the variable $x = (x_1, \dots, x_{deg_d}) \in \mathbb{R}^{deg_d}$ with respect to the $\ell^1$-norm, that is
    \begin{equation} \label{eq:generalized.Lipschitz.assumption} | G_d(v, x_1, \dots, x_{deg_d} ) - G_d(v, x'_1, \dots x'_{deg_d}) | \leq L_{all} \sum_{i = 1}^{deg_d} | x_i - x'_i |. \end{equation}
    If $type(d-1) = non-det$, we assume that $G_d(v, \cdot)$ is $L$-Lipschitz continuous with respect to the $\ell^1$-norm.
\end{assumption}

\begin{remark} \label{rem.l1.vs.l2}
    Since all $\ell^p$-norms are equivalent, the assumption above is equivalent to assuming that the $G_d$ are Lipschitz-continuous with respect to the $\ell^2$-norm. However, the Lipschitz constant with respect to the $\ell^2$-norm will be different in general. 
    As we will see, the complexity of our algorithm significantly improves if the Lipschitz constant is $\leq 1$, so we do care about the precise value of this constant.
\end{remark}

\begin{remark} \label{rem:lipschitz.constant.1}
    The complexity bounds below behave much better if $deg L = 1$. This occurs in various \mrn{situations, for example in expectiminimax algorithms, where the non-deterministic steps all have the mean as their $G_d$, which is $\frac{1}{deg}$-Lipschitz continuous.}
    Because \mrn{the sample complexity of the algorithm improves with lower L, it is important to upper bound the Lipschitz constants of the value functions $G_d$. Factors of $L$ enter the complexity} \grn{through $\kappa_d = deg\,L$ in instance (a) of \vio{Example} \ref{rem:det.instances} and through $\Lambda_d \leq deg\,L$ in Assumption \ref{ass:non-det.value.passing.is.linear}; for $\max$, $\min$ and the mean neither bound is needed.}
\end{remark}

\subsection{Complexity measure and the classical baseline} \label{sec:complexity.measure}

\grn{We count queries to a value oracle $O_x : \ket{v}\ket{0} \mapsto \ket{v}\ket{x_v}$ on the leaves, and we ask for an estimate of \mrn{$\gamma_1(root)$} with root mean square error (RMSE) at most $\epsilon$. The leaf values are the only problem-specific input. In a game tree they come from a static evaluation function whose cost dwarfs the $O(\log(N/\epsilon))$-qubit arithmetic at the internal nodes, so the number of leaf evaluations is the right measure.}

\grn{For the classical comparison write $D = 2m$: the tree has $m$ deterministic layers alternating with $m$ chance layers, and $deg^{2m}$ leaves. Exact evaluation reads every leaf. If RMSE $\epsilon$ is enough, a classical algorithm can run the same derandomised multilevel scheme we use, paying the classical mean-estimation cost $\epsilon^{-2}$ instead of $\epsilon^{-1}$, but it still touches all $deg$ children of every deterministic node. This gives the baseline
\begin{equation} \label{eq:classical.baseline}
    Q^{\mathrm{cl}}_\epsilon(T) \; = \; \widetilde O\!\left( deg^{m} \frac{1}{\epsilon^2} \right) ,
\end{equation}
where the exponent $m$ counts only the deterministic layers. The chance layers are sampled, not evaluated, and they contribute the dependence on $\epsilon$. 

\mrn{Classical algorithms can use pruning techniques to avoid evaluating the whole tree in the average case. This can be used to improve the randomized sample complexity of classical algorithms to $O(deg^{cD})$ with some constant $c<1$~\cite{saksProbabilisticBooleanDecision1986}.}
As far as we know, no pruning technique for approximately evaluated expectiminimax trees exists \mrn{that results in a reduced sample complexity. Furthermore, in this paper we analyze the worst-case complexity of the algorithm, which is unaffected by any known pruning technique in this setting. Because of this, we disregard pruning techniques in the rest of this paper.}}

\grn{Both parameters of \eqref{eq:classical.baseline} are necessary classically. We will discuss this in more detail in Section \ref{sec:optimality} together with lower bounds on the query-complexity of quantum algorithms. As we will see, the bound established in Theorem \ref{thm:headline} beats every classical algorithm in each parameter with the other held fixed.} 

\section{Preliminaries} \label{sec:prelim}
The following is a corollary of Cauchy--Schwarz \grn{(from \cite[Lemma 2.3]{sunOptimalQuantumSpeedups2026})} that we will use.

\begin{lemma} \label{lem:CS.application}
    For any random variables $X_1, \dots, X_n$, we have
    \[ \mathbb{E}[ (X_1 + \dots + X_n)^2] \leq n \sum_{i=1}^n \mathbb{E}[X_i^2]. \]
\end{lemma}

\grn{For maxima we use the following.}

\begin{lemma} \label{lem:max.application}
    \grn{For any random variables $Y_1, \dots, Y_n$, we have
    \[ \mathbb{E}\left[ \left(\max_{i} Y_i\right)^2 \right] \leq \mathbb{E}\left[ \max_i Y_i^2 \right] \leq \sum_{i=1}^n \mathbb{E}[Y_i^2]. \]}

    \grn{\[ \mathbb{E} \left[ \left( \min_i Y_i \right)^2 \right] \leq \mathbb{E}\left[ \max_i Y_i^2 \right] \leq \sum_{i=1}^n \mathbb{E}[Y_i^2]. \]}
\end{lemma}

\begin{proof}
    \grn{If $M = \max_i Y_i$ then $|M| \leq \max_i |Y_i|$, so $M^2 \leq \max_i Y_i^2 \leq \sum_i Y_i^2$. Taking expectations gives the first chain.} \grn{For the minimum, $m = \min_i Y_i$ is one of the $Y_i$, say $m = Y_k$, so $|m| = |Y_k| \leq \max_i |Y_i|$ and $m^2 \leq \max_i Y_i^2 \leq \sum_i Y_i^2$. Taking expectations gives the second chain. The middle term must be $\max_i Y_i^2$ in both chains: $\min_i |Y_i|$ does not bound $|\min_i Y_i|$ (take $Y = (-3, 1)$, where $|\min_i Y_i| = 3$ and $\min_i |Y_i| = 1$), so $\mathbb{E}[\min_i Y_i^2]$ is a lower bound for $\mathbb{E}[(\min_i Y_i)^2]$, not an upper bound.}
\end{proof}

\grn{A uniform hypothesis gives another inequality. Take $|Y_i| \leq \epsilon$ with probability $1 - \eta$ and $|Y_i| \leq M$ \grn{always}. Condition on whether any of the $n$ variables fails and apply a union bound:}
\begin{equation} \label{eq:max.highprob}
    \grn{\mathbb{E}\left[ \max_i Y_i^2 \right] \leq \epsilon^2 + n \eta M^2,}
\end{equation}
\grn{so $\eta = \epsilon^2/(n M^2)$ gives $\mathbb{E}[\max_i Y_i^2] \leq 2\epsilon^2$ \grn{from a hypothesis on each $Y_i$ alone. The hypothesis is uniform, not mean-square, so a child with RMSE $\epsilon$ does not supply it for free. By Lemma \ref{lem:median.bridge} below, a child computed to RMSE $\frac{\epsilon}{2}$ and boosted by a reversible median over $O(\log(n M/\epsilon))$ repetitions \cite{nagajFastAmplificationQMA2009,rallFasterCoherentQuantum2021} satisfies $|Y_i| \leq \epsilon$ with probability $1-\eta$, which is what (\ref{eq:max.highprob}) needs.}}

\grn{We use the boost in two places, and in both a reversible circuit computes the median and leaves it in place. Where the boosted value feeds a mean estimator this is allowed because a synthesizer may carry garbage \cite[Def.~2.2 and Rem.~2.6]{kothariMeanEstimationSourceCode2023}, and the analysis of \cite{rallFasterCoherentQuantum2021} focuses on the output register. The amplitude amplification of $U_{step}$ in Algorithm \ref{alg:min-max-binary-search} queries the output\grn{. This is possible because} amplitude amplification requires a unitary $\mathcal{A}$ and a projector $P$. It rotates in the plane spanned by $P\mathcal{A}\ket{0}$ and $(1-P)\mathcal{A}\ket{0}$ \cite{brassardQuantumAmplitudeAmplification2002}. Here $\mathcal{A}$ is the uniform superposition over the children followed by the boosted estimator and the comparison, $P$ projects onto comparison bit $1$, and the reflection $\mathcal{A}S_0\mathcal{A}^\dagger$ \vio{about the state $\mathcal{A}\ket{0}$, with $S_0 := \mathrm{Id} - 2\ket{0}\!\bra{0}$,}  runs the boosted estimator backwards\grn{, garbage included}. Lemma \ref{lem:search.step} is stated for this garbage-carrying oracle.}

\begin{remark} \label{rem:max.lemma.is.tight}
    \grn{Let $Y_1, \dots, Y_n$ be independent with $Y_i = \sqrt{n}\,\epsilon$ with probability $\frac{1}{n}$ and $Y_i = 0$ otherwise. Then $\mathbb{E}[Y_i^2] = \epsilon^2$ for every $i$, while}
    \[ \grn{\mathbb{E}\left[ \max_i Y_i^2 \right] \; = \; n \epsilon^2 \left( 1 - \left( 1 - \tfrac{1}{n} \right)^{n} \right) \; \geq \; \left( 1 - e^{-1} \right) n \epsilon^2.} \]
    \grn{So a mean-square hypothesis on the children forces a loss of $\sqrt{n}$ at an extremum, and no rearrangement of the estimate avoids it. The uniform bound of (\ref{eq:max.highprob}) does avoid it. This is why an extremum over mean-square estimates must lose $\sqrt{deg}$, and why the deterministic step has to strengthen what a child brings up.}
\end{remark}

\grn{The two hypotheses \mrn{ -- RMSE bound and uniform bound with some probability --} convert into each other in both directions at the cost of one logarithm. Algorithm \ref{alg:min-max-binary-search} needs this for its inputs.}

\begin{lemma}[Median bridge] \label{lem:median.bridge}
    \grn{Let $R$ be an estimator of $\gamma$ with $\mathbb{E}[(R-\gamma)^2] \leq \epsilon_0^2$, realised by a unitary, and let $\bar R$ be the median of $r$ independent copies of $R$, computed by a reversible median circuit. Then}
    \[ \grn{\mathbb{P}\big( |\bar R - \gamma| > 2\epsilon_0 \big) \; \leq \; e^{-r/8},} \]
    \grn{so $r = 8\lceil\log(1/\eta)\rceil$ copies give accuracy $2\epsilon_0$ with failure probability at most $\eta$. Conversely, if $|R-\gamma| \leq \epsilon$ with probability at least $1-\eta$ and $|R-\gamma| \leq M$ always, then $\mathbb{E}[(R-\gamma)^2] \leq \epsilon^2 + \eta M^2$.}
\end{lemma}

\begin{proof}
    \grn{By Chebyshev, $\mathbb{P}(|R^{(i)}-\gamma| > 2\epsilon_0) \leq \epsilon_0^2/(4\epsilon_0^2) = 1/4$ for each copy. The median exceeds $2\epsilon_0$ in absolute value only if at least $r/2$ of the copies do, and by Hoeffding's inequality this has probability at most $\exp(-2r(1/2-1/4)^2) = e^{-r/8}$. The converse is (\ref{eq:max.highprob}) with $n = 1$.}
\end{proof}

\grn{Chebyshev gives failure probability only $\frac14$ at deviation $2\epsilon_0$, so a child asked for uniform accuracy $\epsilon_0$ must be computed to RMSE $\frac{\epsilon_0}{2}$. Over the $\frac{D}{2}$ deterministic layers this compounds to a factor $2^{D/2}$ in the leaf accuracy, one source of the constant $K^{D}$ in the bounds below. \grn{The binary search of Section \ref{sec:det.step.idea} needs bounded-error inputs, the chance steps deliver RMSE, and this lemma converts one into the other at the cost of one logarithm per deterministic layer.}}

\grn{Using the uniform guarantee, an extremum over $deg$ children loses only a constant. Corollary \ref{cor:unconditional} uses the following lemma. The binary search of Section \ref{sec:det.step.idea} achieves the same without it, by bracketing the true extremum directly.}

\begin{lemma}[Deterministic step without loss in $deg$] \label{lem:det.step.no.deg}
    \grn{Let $v \in V(T)_{d-1}$ with $type(d-1) = det$ and $G_d \in \{\max, \min\}$, let $\epsilon_0 > 0$, and assume $|R_{d+1}(v',\epsilon_0) - \gamma_{d+1}(v')| \leq N$ always. Let $\bar R_{d+1}(v',\epsilon_0)$ be the median of $r:= 8 \lceil \log(deg\, N^2/\epsilon_0^2) \rceil$ independent copies of $R_{d+1}(v',\epsilon_0)$, computed by a reversible median circuit. Then}
    \[ \grn{\mathbb{E}\left[ \left( G_d(v, \bar R_{d+1}(v',\epsilon_0)) - G_d(v, \gamma_{d+1}(v')) \right)^2 \right] \; \leq \; 5 \epsilon_0^2,} \]
    \grn{at a cost of $r = O(\log(deg\,N/\epsilon_0))$ times the cost of one $R_{d+1}(v',\epsilon_0)$ per child.}
\end{lemma}

\begin{proof}
    \grn{Put $Y_{v'}:= \bar R_{d+1}(v',\epsilon_0) - \gamma_{d+1}(v')$. Lemma \ref{lem:median.bridge} gives $\mathbb{P}(|Y_{v'}| > 2\epsilon_0) \leq e^{-r/8} \leq \frac{\epsilon_0^2}{deg\,N^2}$ for each of the $deg$ children, and $|Y_{v'}| \leq N$ always. Since $|\max_i a_i - \max_i b_i| \leq \max_i |a_i - b_i|$, and likewise for $\min$, the left-hand side is at most $\mathbb{E}[\max_{v'} Y_{v'}^2]$. Now (\ref{eq:max.highprob}) with $\epsilon = 2\epsilon_0$, $\eta = \frac{\epsilon_0^2}{deg\,N^2}$ and $M = N$ bounds this by $4\epsilon_0^2 + deg \cdot \frac{\epsilon_0^2}{deg\,N^2} \cdot N^2 = 5\epsilon_0^2$.}
\end{proof}

\section{The algorithm as a recursive family of unitaries} \label{sec:generalized.algorithm}

{\authorcolour{teal}

\subsection{Assumptions on the layer transitions} \label{subsec:general.algorithm.with.unitaries}

Our core algorithm is about taking individual algorithms that compute the transition from layer $d+1$ to layer $d$ and putting them together efficiently. We are thus assuming that there exists an efficient algorithm for every layer transition. In Section \ref{sec:unitary} we give a concrete class of examples where these efficient transition-algorithms exist. \mrn{We highlight that we expect several other classes where our algorithm can be utilized.}

\grn{In this subsection and the next, every algorithm is a unitary circuit, and we fix the following conventions. A unitary ``at layer $d$'', written $U_d$ or $U_{d,\epsilon}$, acts on three registers: a vertex register holding a basis state $\ket{v}$ with $v \in V(T)_{d-1}$, an output register of $b = O(\log(1/\epsilon))$ qubits \grn{($O(\log(N/\epsilon))$ before the normalization fixed below)} whose basis states we read as fixed-point numbers, and a workspace register that absorbs everything else the circuit produces (the garbage). For a basis state $\ket{v}$ we write $R_d(v)$ for the random variable we would get by measuring the output register of $U_d\ket{v}\ket{0}\ket{0}$ in the computational basis and discarding the workspace. This measurement only names a distribution; we never perform it below the top level (Section \ref{sec:unitary.formulation} gives the precise form). By Definition \ref{def:no.blowup} we shift and rescale the leaf values once and for all so that every $\gamma_d(v)$ lies in $[0,1]$; an RMSE of $\epsilon$ in the original units is an RMSE of $\epsilon/\grn{(CN)}$ after this rescaling. Finally, every unitary we construct clips its output register to $[0,1]$ by reversible arithmetic before returning. All targets lie in $[0,1]$, so clipping never increases $|R_d(v) - \gamma_d(v)|$, and it gives the crude bound $|R_d(v) - \gamma_d(v)| \leq 1$ that the deterministic step needs for its failure event.} \mrn{ We recall that, by Assumption \ref{ass:generalized.Lipschitz}, all $G_d$ are $L$-Lipschitz continuous.}

\begin{assumption} \label{ass:non-det.value.passing.is.linear}
    If $type(\grn{d-1}) = non\text{-}det$, we assume that $G_{\grn{d}}(v, x_1, \dots, x_{deg})$ is linear in the vector $(x_1, \dots, x_{deg})$\grn{, that is,
    \[ G_d(v, x_1, \dots, x_{deg}) = \sum_{i=1}^{deg} c_i(v)\, x_i \]
    for known coefficients $c_i(v)$. We write $\Lambda_d:= \max_{v \in V(T)_{d-1}} \sum_{i=1}^{deg} |c_i(v)|$. The $\ell^1$-Lipschitz constant of a linear functional is $\max_i |c_i(v)|$, so $\Lambda_d \leq deg \cdot L$ always}.
\end{assumption}

An important example that satisfies assumption \ref{ass:non-det.value.passing.is.linear} is when $G_{\grn{d}}$ is the mean, that is
\[ G_{\grn{d}}(v, x_1, \dots, x_{deg}) = \sum_{i = 1}^{deg} \frac{x_i}{deg}\grn{, \qquad \Lambda_d = 1}. \]
\grn{By the normalization fixed above, all values $\gamma_d(v)$ lie in $[0,1]$.}

\begin{assumption} \label{ass:deterministic.step.unitary}
    If $type(d-1) = det$, there exists an algorithm denoted $ConstructUnitary\_d$ with the following properties:
    \begin{enumerate}
        \item The input of $ConstructUnitary\_d$ is a unitary $U_{d+1}$ \grn{at layer $d+1$, that is, one acting on $V(T)_{d}$, an output and a garbage register,} and a positive number $\epsilon > 0$.

        \item The output is a unitary $U_{d, \epsilon}$ that has three registers \mrn{as outlined above}: One register goes across the vertices $V(T)_{d-1}$; one register is the output register; one register is the garbage register.

        \item For $v \in V(T)_{d-1}$, let $R_d(v, \epsilon)$ denote the random variable obtained from measuring the output register of $U_{d, \epsilon} \ket{v} \ket{0} \ket{0}$ and $R_{d+1}(v')$ denote the random variable obtained from measuring the output register of $U_{d+1} \ket{v'} \ket{0} \ket{0}$.
        \grn{A constant $\kappa_d \geq 1$ belongs to $ConstructUnitary\_d$, with the following property. Suppose there are numbers $(y_{v'})_{v' \in V(T)_d}$ in $[0,1]$ such that for every $v' \in V(T)_d$
        \[ \mathbb{E}\big[(R_{d+1}(v') - y_{v'})^2\big] \leq \Big(\frac{\epsilon}{\kappa_d}\Big)^2 \qquad\text{and}\qquad |R_{d+1}(v') - y_{v'}| \leq 1 \text{ always}. \]
        Then for every $v \in V(T)_{d-1}$ with children $w_1, \dots, w_{deg}$,
        \[ \mathbb{E}\big[(R_d(v,\epsilon) - G_d(v, y_{w_1}, \dots, y_{w_{deg}}))^2\big] \leq \epsilon^2. \]}

        \item The unitary $U_{d,\epsilon}$ has to call the unitary $U_{d+1}$ \grn{(or its inverse, or a controlled version)} at most \grn{$q_d(\epsilon)$} many times\grn{, where $q_d$ is a non-increasing function of $\epsilon$ that belongs to $ConstructUnitary\_d$. In the instances we care about, $q_d(\epsilon) = O(deg^{\alpha} \cdot polylog(deg, \frac{1}{\epsilon}))$, with $\alpha = \frac12$ for a search-based step and $\alpha = 1$ for a step that evaluates all children} \vio{Here $polylog(deg, \frac{1}{\epsilon})$ stands for a product of powers of $\log (deg)$ and $\log(\frac{1}{\epsilon})$ whose degree is fixed by the subroutine.}.
    \end{enumerate}
\end{assumption}

\begin{example}[Instances of Assumption \ref{ass:deterministic.step.unitary}] \label{rem:det.instances}
    \grn{(a) \emph{Evaluate all children.} For any $G_d$ that is $L$-Lipschitz with respect to $\ell^1$, the unitary that applies $U_{d+1}$ to each child on its own workspace and computes $G_d$ by reversible arithmetic satisfies the assumption with $q_d = deg$ and $\kappa_d = deg\,L$. By (\ref{eq:generalized.Lipschitz.assumption}) and Lemma \ref{lem:CS.application},
    \[ \mathbb{E}\big[(G_d(v,R_{d+1}(w)) - G_d(v,y_w))^2\big] \leq L^2\, deg \sum_{w} \mathbb{E}\big[(R_{d+1}(w)-y_w)^2\big] \leq deg^2 L^2 \Big(\frac{\epsilon}{\kappa_d}\Big)^2 = \epsilon^2. \]
    This instance is always available and gives no speedup in $deg$.
    (b) \emph{Evaluate all children, extremum with median boost.} For $G_d \in \{\max,\min\}$, replace each child by the median of $r = O(\log(deg/\epsilon))$ independent copies before taking the extremum. This gives $\kappa_d = \sqrt{5}$ and $q_d = deg \cdot r$; it is Lemma \ref{lem:det.step.no.deg}. Without the boost, Lemma \ref{lem:max.application} gives $\kappa_d = \sqrt{deg}$ and $q_d = deg$, and by Remark \ref{rem:max.lemma.is.tight} this $\sqrt{deg}$ cannot be avoided under a mean-square hypothesis.
    (c) \emph{Search.} For $G_d \in \{\max,\min\}$, Algorithm \ref{alg:min-max-binary-search} satisfies the assumption with $\kappa_d = 8$ and $q_d(\epsilon) = O\big(\sqrt{deg}\,\log(\tfrac{1}{\epsilon})\log(\tfrac{\log(1/\epsilon)}{\epsilon})\log (deg) \big)$. This is Lemma \ref{lem:minmax.satisfies.assumption}, and it is the instance that gives the quadratic speedup in $deg$.}
\end{example}

Assumption \ref{ass:deterministic.step.unitary} above covers how we will handle the deterministic steps. The assumption below deals with the non-deterministic step. As we will show right after, the assumption below is in fact always satisfied. We still phrase it as an assumption because it gives greater clarity on what our algorithm uses and it illustrates what one has to do for Assumption \ref{ass:deterministic.step.unitary} when building new applications of our algorithm.

\begin{assumption} \label{ass:non-deterministic.step.unitary}
    If $type(\grn{d-1}) = non\text{-}det$ \grn{(so $G_d$ is linear by Assumption \ref{ass:non-det.value.passing.is.linear})}, there exists an algorithm denoted $ConstructUnitary\_non\_det\_\grn{d}$ with the following properties:
    \begin{enumerate}
        \item The input of $ConstructUnitary\_non\_det\_\grn{d}$ is a unitary $U_{\Delta}$ \grn{at layer $d+1$, that is, one acting on $V(T)_{d}$, an output and a garbage register,} \grn{a number $s > 0$,} and a positive number $\epsilon > 0$.

        \item The output is a unitary $A_{\grn{d}, \epsilon}$ at layer $d$, that is: One register goes across the vertices $V(T)_{\grn{d-1}}$; one register is the output register; one register is the garbage register.

        \item For $v \in V(T)_{\grn{d-1}}$, let $A_{\grn{d}, \epsilon}(v)$ denote the random variable obtained by measuring the output register of $A_{\grn{d}, \epsilon} \ket{v} \ket{0} \ket{0}$. Let $\Delta(w)$ denote the random variable obtained by measuring the output register of $U_{\Delta}$ when applied to $\ket{w} \ket{0} \ket{0}$ for $w \in V(T)_{\grn{d}}$.

        \grn{Suppose that $\mathbb{E}[\Delta(w)^2] \leq s^2$ for every $w \in V(T)_d$. Then for every $v \in V(T)_{d-1}$ with children $w_1, \dots, w_{deg}$,
        \[ \mathbb{E}\Big[\big(A_{d,\epsilon}(v) - G_d\big(v, \mathbb{E}[\Delta(w_1)], \dots, \mathbb{E}[\Delta(w_{deg})]\big)\big)^2\Big] \leq \epsilon^2. \]}

        \item The unitary $A_{\grn{d},\epsilon}$ has to call the unitary $U_{\Delta}$ \grn{(or its inverse, or a controlled version)} at most \grn{$c_Q \cdot \frac{\Lambda_d\, s}{\epsilon}\log\big(\frac{\Lambda_d\, s}{\epsilon}\big)$} many times\grn{, where $c_Q$ is an absolute constant}.
    \end{enumerate}
\end{assumption}

\begin{lemma} \label{lem:non.det.assumption.is.satisfied}
    If $type(d-1) = non-det$ and $G_d$ is linear, then Assumption \ref{ass:non-deterministic.step.unitary} is always satisfied with $s = \sup_{w \in V(T)_{d}} \left( \sqrt{ \mathbb{E}[\Delta(w)^2]} \right)$. The algorithm is an application of the Algorithm in {\cite[Corollary 3.2]{sunOptimalQuantumSpeedups2026}}.
\end{lemma}

\begin{proof}
    Since $G_d$ is a linear map $\mathbb{R}^{deg} \rightarrow \mathbb{R}$, we can write
    \[ G_d(v, x_1, \dots, x_{deg}) = \sum_{i=1}^{deg} c_i(v) x_i. \]
    Set $s := \sup_{w \in V(T)_{d}} \left( \sqrt{ \mathbb{E}[\Delta(w)^2]} \right)$. Now let $X$ be the random variable defined as follows: It chooses an index $J \in \{ 1, \dots, deg \}$ with probability $\frac{ \vert c_J(v) \vert}{\Lambda_v}$, where $\Lambda_v:= \sum_i |c_i(v)| \leq \Lambda_d$ and, conditional on the choice of $J$, we set
    \[ X := \Lambda_v \mathrm{sgn}(c_J(v)) \Delta(w_J). \]
    This random variable satisfies
    \[ \mathbb{E}[X] = \sum_{J=1}^{deg} \frac{ \vert c_J(v) \vert}{\Lambda_v} \Lambda_v \mathrm{sgn}( c_J(v) ) \mathbb{E}[\Delta(w_J)] = \sum_{J=1}^{deg} c_J(v) \mathbb{E}[\Delta(w_J)] = G_d(v, \mathbb{E}[\Delta(w)]), \]
    \[ \mathbb{E}[X^2] = \sum_{J=1}^{deg} \frac{ \vert c_J(v) \vert}{\Lambda_v} \Lambda_v^2 \mathrm{sgn}( c_J(v) )^2 \mathbb{E}[\Delta(w_J)^2] \leq \sum_{J=1}^{deg} \vert c_J(v) \vert \Lambda_v s^2 \leq \Lambda_d^2 s^2. \]
    The random variable $X$ has a synthesizer given by $U_{\Delta} U_c$, where $U_c \ket{v} \ket{0} = \sum_{J=1}^{deg} \sqrt{\frac{\vert c_J(v) \vert}{\Lambda_v}} \ket{v} \ket{w_J}$, followed by reversible multiplication of the output by $\Lambda_v \mathrm{sgn}(c_J(v))$. We now apply the algorithm of {\cite[Cor. 3.2]{sunOptimalQuantumSpeedups2026}} with second-moment bound $\Lambda_d s$, which gives us an approximation of $\mathbb{E}[X] = G_d(v, \mathbb{E}[\Delta(w)])$ with the level of accuracy required by point 3 and the query-complexity required by point 4 in Assumption \ref{ass:non-deterministic.step.unitary}. Note that the algorithm of {\cite[Cor. 3.2]{sunOptimalQuantumSpeedups2026}} gives a query-complexity of $O(\frac{s'}{\epsilon} \log(\frac{s'}{\epsilon})$, where $s'^2 \geq \mathbb{E}[X^2]$ for the random variable $X$ whose mean is approximated. In our case $s' = \Lambda_d s$, which is why the $\Lambda_d$ appears in our bound of the query-complexity.
\end{proof}

Note that the dependence of the algorithm constructed in Lemma \ref{lem:non.det.assumption.is.satisfied} of the function $G_d$ is expressed fully in terms of the $c_i(v)$. In other words, we can extract an algorithm from the Lemma which takes the weights $c_i(v)$ as an input and it outputs an algorithm $ConstructUnitary\_non\_det\_d$ that satisfies Assumption \ref{ass:non-deterministic.step.unitary} for the function $G_d$ that is defined by the weights $c_i(v)$ and $s$ as in the Lemma.

\mrn{The algorithm given in {\cite[Cor 3.2]{sunOptimalQuantumSpeedups2026}} requires that one has access to "the code of $X$", where $X$ is the random variable from the proof above. We will discuss in Appendix \ref{app:coherent.composition} what this does and does not require.}

Under these assumptions, consider the following \grn{algorithm, which is} a recursive construction of one unitary.

\begin{algorithm}
\caption{BuildUnitaryEstimator(d, $\epsilon$) \grn{(abstract version)}} \label{alg:generalized.unitary.estimator.constructor}

\begin{algorithmic}[1]
    \If{d = D + 1}
        \State Define the unitary $U_{d,\epsilon}$ by $U_{d, \epsilon} \ket{v} \ket{0} \ket{0}:= \ket{v} \ket{val(v)} \ket{0}$
        \State \Return $U_{d,\epsilon}$
    \EndIf

    \If{$type(d-1) = det$} \Comment{$U_{d,\epsilon}$ acts on vertices of depth $d-1$; the step to their children has type $type(d-1)$.}

        \State $U_{d+1} \gets BuildUnitaryEstimator(d+1,\grn{\epsilon/\kappa_d})$

        \State $U_{d, \epsilon} \gets ConstructUnitary\_d(U_{d+1}, \epsilon)$\grn{, followed by clipping the output register to $[0,1]$}

        \State \Return $U_{d,\epsilon}$

    \EndIf

    \If{$type(d-1) = non\text{-}det$}
        \State $B_{\grn{d}} \gets \grn{\max\Big(0,\; \grn{\Big\lceil} 2 \log_2\Big( \frac{ 2 \sqrt{2}\, deg L }{\epsilon} \Big) \grn{\Big\rceil}\Big)}$
        \For{$0 \leq n \leq B_{\grn{d}}$}
            \State $U_{d+1,n} \gets BuildUnitaryEstimator(d+1, 2^{-\frac{n}{2}})$

            \State $U_{\Delta,n} \gets$ the unitary that applies $U_{\grn{d+1},n}$ and $U_{\grn{d+1},n-1}$ on separate workspaces, unitarily subtracts the result of $U_{\grn{d+1}, n-1}$ from the result of $U_{\grn{d+1},n}$ and puts this in the second register, putting all garbage and intermediate outputs in the last register. \Comment{If $n=0$, we simply set $U_{\Delta,0} \gets U_{\grn{d+1},0}$.}

            \State \grn{$s_n \gets 3 \cdot 2^{-n/2}$} \Comment{\vio{second-moment bound for $U_{\Delta,n}$: $\mathbb{E}[\Delta_n(w)^2] \leq s_n^2$ by (\ref{eq:unitary.second.momentum})}}

            \State $A_{d,n} \gets ConstructUnitary\_non\_det\_d\big(U_{\Delta,n}, \grn{s_n,}\ \grn{\frac{\epsilon}{3(B_d+1)}}\big)$


        \EndFor
        \State $U_{d, \epsilon}$ is defined as the unitary that takes a $\ket{v}, v \in V(T)_{\grn{d-1}}$ as input and for every $n$, applies $A_{d,n}$ to $\ket{v}$ on a separate workspace\grn{,}  coherently computes their sum in the output register\grn{, and clips it to $[0,1]$}. 

        \State \Return $U_{d, \epsilon}$

    \EndIf
\end{algorithmic}
\end{algorithm}

\subsection{Estimating the RMSE and query complexity} \label{subsec:induction.steps}

\begin{theorem}[Guarantee for Algorithm \ref{alg:generalized.unitary.estimator.constructor}] \label{thm:unitary.general}
    \grn{Suppose $T$ and $(G_d)_d$ have no linear blowup and are normalised so that all $\gamma_d(v) \in [0,1]$, and suppose Assumptions \ref{ass:non-det.value.passing.is.linear}, \ref{ass:deterministic.step.unitary} and \ref{ass:non-deterministic.step.unitary} hold, with $\Lambda_j \geq 1$ at every non-deterministic step. \mrn{Suppose in addition that $q_j(\epsilon) = deg^{\alpha} \; polylog(deg, \frac{1}{\epsilon})$, where $polylog(deg, \frac{1}{\epsilon})$ is a polynomial of degree $\beta$.} \vio{Recall the three parameters. At a deterministic step $j$, $\kappa_j \geq 1$ is the factor by which the children must be more accurate than the parent and $q_j(\cdot)$ is the number of calls to the child unitary (Assumption \ref{ass:deterministic.step.unitary}, points 3 and 4). At a non-deterministic step $j$, $\Lambda_j = \max_{v} \sum_i |c_i(v)|$ is the $\ell^1$ norm of the coefficients of the linear $G_j$ (Assumption \ref{ass:non-det.value.passing.is.linear}), equal to $1$ for any probability vector.} Let $0 < \epsilon \leq 1$ and put} $B_j:= \grn{\Big\lceil} 2\log_2\Big(\frac{2\sqrt2\,deg L}{\epsilon_{\min\vio{,j}}}\Big) \grn{\Big\rceil}$.
    Let $U_{d,\epsilon} = BuildUnitaryEstimator(d,\epsilon)$ and, for $v \in V(T)_{d-1}$, let $R_d(v,\epsilon)$ be the random variable obtained by measuring the output register of $U_{d,\epsilon}\ket{v}\ket{0}\ket{0}$. Then
    \[ \mathbb{E}[(R_d(v,\epsilon) - \gamma_d(v))^2] \leq \epsilon^2, \]
    and the number $Q(U_{d,\epsilon})$ of applications of the leaf oracle (or its inverse or controlled versions) inside $U_{d,\epsilon}$ satisfies
    \begin{equation} \label{eq:general.theorem}
        Q_d(U_{d,\epsilon}) = O\left( \frac{1}{\epsilon} deg^{\alpha \frac{D-d}{2}} K^{\frac{D-d}{2}} deg^{\frac{D-d}{2}} L^{\frac{D-d}{2}} \kappa^{\frac{D-d}{2}} (D-d)^{\max(4, \beta)(D-d)} \log\left( \frac{deg L \kappa}{\epsilon}\right)^{\max(4, \beta)(D-d)} \right).
    \end{equation}
\end{theorem}

We present a sketch of the proof here. The full proof is in Appendix \ref{sec:proof.2.of.main.theorem}.

\begin{proof}[Proof sketch]
\grn{We use backward induction on $d$ with hypothesis $P(d)$: for all $\delta \in [\epsilon_{\min},1]$ and $v \in V(T)_{d-1}$, $\mathbb{E}[(R_d(v,\delta)-\gamma_d(v))^2] \leq \delta^2$ and $Q(U_{d,\delta}) \leq M_d/\delta$. At a deterministic layer, $P(d+1)$ at accuracy $\delta/\kappa_d$ is the hypothesis of point 3 of Assumption \ref{ass:deterministic.step.unitary} with $y_{v'} = \gamma_{d+1}(v')$. This gives the RMSE bound, and point 4 gives $M_d = q_d(\epsilon_{\min})\,\kappa_d\, M_{d+1}$. At a non-deterministic layer, $P(d+1)$ gives $\mathbb{E}[\Delta_n(w)^2] \leq 2\cdot 2^{-n} + 2\cdot 2^{-(n-1)} = 6 \cdot 2^{-n} \leq s_n^2$, because both estimators inside $\Delta_n$ are compared to the common target $\gamma_{d+1}(w)$ and never to each other. The expectations telescope, $\sum_{n \leq B_d} \mathbb{E}[\Delta_n(w)] = \mathbb{E}[R_{d+1}(w,2^{-B_d/2})]$, so by linearity of $G_d$ the targets of the $A_{d,n}$ sum to $G_d(v,\mathbb{E}[R_{d+1}(w,2^{-B_d/2})])$, which is within $\Lambda_d 2^{-B_d/2}$ of $\gamma_d(v)$ by Jensen's inequality. Lemma \ref{lem:CS.application} then bounds the mean square error by $2(B_d+1)^2\epsilon'^2 + 2\Lambda_d^2 2^{-B_d} = (\frac29 + \frac14)\delta^2$. For the cost, level $n$ makes $O(\Lambda_d s_n (B_d+1)\delta^{-1}\log(\cdot))$ calls to $U_{\Delta,n}$, each costing $O(M_{d+1}2^{n/2})$. The factor $2^{-n/2}$ in $s_n$ cancels the $2^{n/2}$, so all $B_d+1$ levels cost the same and $M_d = 16 c_Q \Lambda_d (B_d+1)^2 \log(9\Lambda_d(B_d+1)/\epsilon_{\min})\, M_{d+1}$. Unwinding the recursion gives the product.}
\end{proof}

}

\mrn{The proof of Theorem \ref{thm:unitary.general} relies on the fact that we compose the recursive process coherently, that is, the process that passes from any layer to the next one involves no measurement and that for any layer $d-1$, we build one unitary $U_{d,\epsilon}$ that yields a good estimator for every vertex $v \in V(T)_{d-1}$. The fact that this works is the content of the following Lemma, which we prove in Appendix \ref{app:coherent.composition}. 

\begin{lemma}[Coherent composition] \label{lem:coherent.composition}
    \grn{Let $\mathcal{C}$ be a quantum circuit on three registers $\ket{v}\ket{\omega}\ket{g}$: a vertex register, an \emph{output} register whose computational basis states $\ket{\omega}$ we read as fixed-point numbers, and a workspace register. So $\omega$ ranges over the finite grid $\Omega_b \subset \mathbb{R}$ of $b$-bit fixed-point numbers, $X_v$ below is real-valued, and $\omega$ is a number produced by the arithmetic inside $\mathcal{C}$, not the value of a vertex. For a basis state $\ket{v}$ write}
    \[ \grn{\mathcal{C}\ket{v}\ket{0}\ket{0} \; = \; \ket{v} \sum_{\omega} \sqrt{p^{(v)}_\omega}\, \ket{\omega}\ket{g^{(v)}_\omega}, \qquad \lVert g^{(v)}_\omega \rVert = 1,} \]
    \grn{where the sum runs over the basis states of the output register and $\ket{g^{(v)}_\omega}$ collects everything else the circuit produced. Define $X_v$ as the random variable with $\mathbb{P}[X_v = \omega] = p^{(v)}_\omega$, that is, the outcome of measuring the output register in the computational basis and reading it as a number. The measurement is a device for naming the distribution whose mean we estimate and not a step of the algorithm: $\mathcal{C}$ is never measured, and (i) below forbids it. In the language of \cite[Def.~2.2 and 2.10]{kothariMeanEstimationSourceCode2023}, $\mathcal{C}$ is a synthesizer for $p^{(v)}$ and the value map is the identity, the label $\omega$ being the value. Suppose}
    \grn{\begin{enumerate}
        \item[(i)] $\mathcal{C}$ contains no measurement, and controlled-$\mathcal{C}$ and controlled-$\mathcal{C}^\dagger$ are available;
        \item[(ii)] $\mathcal{C}$ is one circuit \mrn{that is} independent of $v$, and so is every parameter of the estimator applied to it;
        \item[(iii)] a bound $s^2 \geq \mathbb{E}[X_v^2]$ is known and holds uniformly in $v$.
    \end{enumerate}}
    \grn{Then the estimator of \cite[Cor.~3.2]{sunOptimalQuantumSpeedups2026} applied to $\mathcal{C}$ is a unitary $U$ with}
    \[ \grn{U \ket{v}\ket{0}\ket{0} \; = \; \ket{v} \sum_{a} \sqrt{q^{(v)}_a}\,\ket{a}\ket{h^{(v)}_a},} \]
    \grn{where $a$ ranges over the basis states of the output register of $U$, read in the same fixed-point convention as $\omega$. (We use a different letter because $U$ has its own output register, holding the estimate of $\mathbb{E}[X_v]$, whereas $\omega$ labels one sample of $X_v$ in the output register of $\mathcal{C}$.) Measuring the output register returns an estimate of $\mathbb{E}[X_v]$ with RMSE at most $\varepsilon$, at a cost of $O\!\left(\frac{s}{\varepsilon}\log\frac{s}{\varepsilon}\right)$ applications of $\mathcal{C}$ and $\mathcal{C}^\dagger$. In particular, the conclusion does not change if $\mathcal{C}$ contains subroutines that are themselves approximate, biased, or occasionally failing.}
\end{lemma}

}

\section{Coherent implementation for stochastic 2-player games} \label{sec:unitary}

\grn{A stochastic 2-player game has exactly the structure of the tree and the passing functions of Section \ref{sec:setting}. It is an expectiminimax tree in the sense of Definition \ref{def:emm.tree} with $\tau$ alternating between a deterministic layer and a chance layer, and from here on we use the two names interchangeably. If $type(d-1)$ is non-deterministic, $G_d$ is the mean,
\[ G_d(v, x_1, \dots, x_{deg}) = \frac{1}{deg} \sum_{i=1}^{deg} x_i , \]
so Assumption \ref{ass:non-det.value.passing.is.linear} holds with $\Lambda_d = deg L = 1$. For the $d$ with $type(d-1)$ deterministic, $G_d$ alternates between the minimum and the maximum,
\[ G_d(v, x_1, \dots, x_{deg}) = \max(x_1, \dots, x_{deg}), \qquad G_d(v, x_1, \dots, x_{deg}) = \min(x_1, \dots, x_{deg}). \]
In this section we give the two layer algorithms\mrn{: First, we give an extremum search for the deterministic layers and we show that it realises $ConstructUnitary\_d$ from Assumption \ref{ass:deterministic.step.unitary}. Then we explain how Lemma \ref{lem:non.det.assumption.is.satisfied} gives us  and $ConstructUnitary\_non\_det\_d$ of Assumption \ref{ass:non-deterministic.step.unitary}. Finally, we derive the main bound.}}

\subsection{The unitary interface} \label{sec:unitary.formulation}

Every step of the recursion is a unitary rather than a sampling procedure. This is what makes the query count well defined, since setting up a unitary over the leaf oracles does not cost any queries\mrn{, only calling it does, and this} is what lets a chance node draw one child in superposition instead of evaluating all of them. A unitary $U_{d, \epsilon}$ applied to a node $v$ gives
\begin{equation}
U_{d, \epsilon} \ket{v} \ket{0} \ket{0} = \sum_{val} \sqrt{p(val)} \ket{v} \ket{val} \ket{garbage_{val}} ,
\end{equation}
such that measuring the $\ket{val}$ register gives the random variable $R_d(v, \epsilon)$ with the properties we describe. The number of times the leaf oracles are queried is then obtained by counting calls of the unitaries, and this setup is what the mean estimator of Kothari and O'Donnell \cite{kothariMeanEstimationSourceCode2023} requires, since it asks for the `code' of the observable.

\grn{We make this precise. Every $U_{d,\epsilon}$ acts on three registers: a vertex register holding a basis state $\ket{v}$ with $v \in V(T)_{d-1}$; an \emph{output} register of $b = O(\log(N/\epsilon))$ qubits whose basis states we read as $b$-bit fixed-point numbers, so that they form a finite grid $\Omega_b \subset \mathbb{R}$; and a workspace register that absorbs everything else the circuit produces. Grouping the amplitude of $U_{d,\epsilon}\ket{v}\ket{0}\ket{0}$ by the content $\omega \in \Omega_b$ of the output register gives the form displayed above, with
\[ p^{(v)}(\omega) \;=\; \bigl\lVert (\mathrm{Id}\otimes\ket{\omega}\!\bra{\omega}\otimes\mathrm{Id})\,U_{d,\epsilon}\ket{v}\ket{0}\ket{0} \bigr\rVert^2, \]
and we \emph{define} $R_d(v,\epsilon)$ as the real-valued random variable with $\mathbb{P}[R_d(v,\epsilon) = \omega] = p^{(v)}(\omega)$: the outcome of measuring the output register in the computational basis and discarding the workspace. Its values are outputs of the arithmetic inside the circuit, a sum of level estimates at a chance layer or the midpoint of a bracket at a deterministic layer. In general they are not values $val(w)$ of vertices; the one exception is the leaf unitary $U_{D+1,\epsilon}$, whose output register holds $val(v)$ exactly. We measure once, at the top level. The algorithm builds the single unitary $U_{1,\epsilon}$ for the root by the recursion of Algorithm \ref{alg:build-unitary-estimator}, in which $U_{d,\epsilon}$ contains as subcircuits the unitaries $U_{d+1,\delta'}$ of its children at the accuracies $\delta'$ it requests, down to the leaf oracle. It applies $U_{1,\epsilon}$ to $\ket{root}\ket{0}\ket{0}$, measures the output register, and returns the outcome as the estimate. Nothing is sampled at the leaves. The leaf values enter only through the oracle $O_x$, queried in superposition, and the only sources of randomness are $U_c$, which places the children of a chance node in superposition, and the coins of the estimators, which are Hadamard-prepared qubits kept in the workspace. Below the top level a subroutine never measures. Where the classical description says ``draw a sample of $R_{d+1}(w,\delta')$'', the circuit applies $U_{d+1,\delta'}$, or its inverse or a controlled version, once. This correspondence lets us reuse the classical analysis: the RMSE of $R_d(v,\epsilon)$ is a functional of $p^{(v)}$ and hence a property of the unitary. Sampling complexity likewise becomes a property of the circuit. $Q(U_{d,\epsilon})$ is the number of applications of $O_x$, $O_x^\dagger$ and their controlled versions inside $U_{d,\epsilon}$. It is a fixed integer, not an expectation over branches, because the circuit is fixed once $d$, $\epsilon$, $deg$ and $N$ are. It obeys the recursion
\[ Q(U_{d,\epsilon}) \;=\; \sum_{\delta'} \bigl(\text{number of applications of } U_{d+1,\delta'} \text{ inside } U_{d,\epsilon}\bigr) \cdot Q(U_{d+1,\delta'}), \]
which Theorem \ref{thm:unitary.general} solves. The query complexity of the whole algorithm is $Q(U_{1,\epsilon})$, since we apply the root unitary once.}

\grn{We can call the child estimators in superposition because of the following construction.\\ $BuildUnitaryEstimator$ builds $U_{d+1}$ as a unitary circuit with every source of randomness purified into a garbage register and no measurement anywhere inside. So applying it to $\sum_j \alpha_j \ket{j}\ket{0}\ket{0}$ is legitimate and returns $\sum_j \alpha_j\ket{j}\ket{\tilde x_j}\ket{g_j}$ with $\tilde x_j$ an estimate of the value of $w_j \mrn{\in V(T)_d}$. What we cannot call in superposition is an estimator that measures internally and branches on the outcome. This is why D\"urr--H\o yer is unavailable to us, and why we state the whole recursion as a family of unitaries rather than a family of sampling procedures.}

{\authorcolour{Maroon}
\subsection{Deterministic layers} \label{sec:det.step.idea}

The deterministic layers need an extremum over $deg$ children with $O(\sqrt{deg})$ queries, the quadratic speedup that Cleve et al.\ realise for minimax trees \cite{cleveQuantumAlgorithmsEvaluating2019}. The known algorithm for finding a minimum with this speedup, due to D\"urr and H\o yer \cite{durrQuantumAlgorithmFinding1999}, cannot be used here for two reasons. It is not coherent: it measures in order to update its pivot, and these measurements destroy the superposition that the mean estimator above and the extremum search of the next layer need. And it assumes children with exact values, whereas our recursive estimators return a different value on every call, with an RMSE guarantee. We therefore use a more naive algorithm, which is also more flexible, at the price of another logarithmic factor: a binary search over the domain $[0,1]$ of the values, fixed by the normalization of Section \ref{sec:setting}, in which each pivot is tested by amplitude amplification over the children. It starts with a pivot of $\frac12$ and moves the pivot up or down depending on whether a child above (below) it is found; it is given in Algorithm~\ref{alg:min-max-binary-search}.

\grn{First we settle the choice of pivot. \textcite{cleveQuantumAlgorithmsEvaluating2019} avoid a midpoint binary search and pivot on a randomly chosen input value instead, because in the input-value model with a large numerical range the midpoint search need not converge in logarithmically many rounds. That objection does not apply here, because we ask only for an additive $\epsilon$ on a range normalised to $[0,1]$. The bracket length obeys $\lambda_{i+1} = \frac{\lambda_i}{2} + \epsilon$, so it reaches $2\epsilon + \epsilon_2$ after $\lceil \log_2(1/\epsilon_2)\rceil$ rounds however the values are distributed. Pivoting on an input value is needed for exact evaluation, not for approximate evaluation on a bounded range.}

\grn{Next we fix which search to use. Querying the children with the bounded-error search of \cite{hoyerQuantumSearchBoundedError2003} would save a logarithmic factor, but that result assumes each $f_j$ is a well-defined bit computed with two-sided error at most $\frac{1}{10}$. Our predicate ``child $j$ lies above the pivot'' has no well-defined value for a child whose true value sits within the estimator's resolution of the pivot. There the comparison returns $1$ with probability near $\frac12$, and no assignment of $f_j$ satisfies the hypothesis. This configuration is not rare: the binary search drives the pivot towards the extremum, and in the last rounds it is the typical one. \textcite{rallFasterCoherentQuantum2021} call this obstruction the rounding promise. So we boost each child until every comparison whose child lies more than $\epsilon$ from the pivot is reliable, and run ordinary amplitude amplification on the result. This is the ``simple search'' that \textcite{hoyerQuantumSearchBoundedError2003} describe in order to improve on it. Lemma \ref{lem:search.step} records why an ambiguous child does no harm. One could also use the search from \cite{gaoQuantumApproximateKMinimum2025}, which treats approximate minimum finding over noisy values directly; we have not checked whether its interface matches ours.}

\begin{lemma}[Correctness of one bracket step] \label{lem:search.step}
    \grn{Let $\gamma_1, \dots, \gamma_{deg} \in [0,1]$, let $\grn{z} \in [0,1]$, and suppose a unitary reads each child and returns a value $\tilde x_j$ with $|\tilde x_j - \gamma_j| \leq \epsilon$ on an event of probability at least $1-\eta$, where $\eta \leq \frac{1}{\grn{C_0}\,deg}$ for an absolute constant $\grn{C_0}$; the read may carry whatever garbage the boost produces. Let \emph{found} be the flag written by the amplitude-amplification procedure described in the proof, and \emph{not found} its negation. Then, except on an event of probability at most $\delta'$,}
    \[ \grn{\text{found} \; \Longrightarrow \; \max_j \gamma_j > \grn{z} - \epsilon, \qquad \text{not found} \; \Longrightarrow \; \max_j \gamma_j \leq \grn{z} + \epsilon.} \]
    \grn{The same statement with the inequalities reversed holds for $\min$ and the predicate $\tilde x_j < \grn{z}$. In particular the bracket update of Algorithm \ref{alg:min-max-binary-search} preserves the invariant $lb \leq \max_j \gamma_j \leq ub$, and no promise on the position of any $\gamma_j$ relative to $\grn{z}$ is needed. The procedure applies the read unitary $O(\sqrt{deg}\,\log(1/\delta'))$ times.}
\end{lemma}

\grn{The proof, in Appendix \ref{app:proof.search}, runs amplitude amplification with a fixed doubling schedule and a majority vote, and checks the three configurations of the true values one by one. A child within $\epsilon$ of the pivot may fall either way without issue.}

\begin{algorithm}
    \caption{MinMaxFindingUnitary($U_{d+1}$, d, $\epsilon$, \grn{$\epsilon_2$, $\delta$, $\mathrm{ext} \in \{\max,\min\}$}) }\label{alg:min-max-binary-search}
    \begin{algorithmic}[1]
        \State $steps \gets \lceil \log_2(1/\grn{\epsilon_2})\rceil$

        \State Define $U_{pivot} \ket{v}\ket{lb}\ket{ub}\ket{0}:= \ket{v}\ket{lb}\ket{ub}\ket{(lb+ub) / 2}$
        \State \grn{Define $U_{child} \ket{v}\ket{j}\ket{0}:= \ket{v}\ket{j}\ket{w_j}$, where $w_j$ is the $j$-th child of $v$}
        \State \grn{Define $U_{est}:= U_{d+1}\, U_{child}$, so that $U_{est}\ket{v}\ket{j}\ket{0}\ket{0}$ carries an estimate of the value of $w_j$ in its last register}
        \State \grn{Define $\bar U_{est}$ as $U_{est}$ followed by a reversible median over $r = 8\lceil\log(\grn{C_0}\,deg)\rceil$ independent copies, as in Lemma \ref{lem:median.bridge}, with the median left in place ($\grn{C_0}$ is the constant of Lemma \ref{lem:search.step}).}
        \State Define $U_{step}$ to be the unitary that takes $\ket{v}\ket{lb}\ket{ub}$ and coherently searches \grn{over $j$} whether one of the estimators given by $\grn{\bar U_{est}}$ is \grn{above} the pivot given by $U_{pivot}$ \grn{if $\mathrm{ext} = \max$, and below it if $\mathrm{ext} = \min$, by amplitude amplification over the $deg$ children as in the proof of Lemma \ref{lem:search.step}}. To reduce the failure rate it does this $\grn{O(\log(steps/\delta))}$ times and uses majority voting to determine if it is found\grn{, so that each round fails with probability at most $\delta/steps$}\grn{ (this majority is the one of Lemma \ref{lem:search.step} with $\delta' = \delta/steps$, not an additional vote)}.  \grn{It then updates the bracket. If $\mathrm{ext} = \max$, controlled on found it changes $\ket{lb}$ to $\ket{\max(lb, pivot-\epsilon)}$, and controlled on not found it changes $\ket{ub}$ to $\ket{\min(ub, pivot+\epsilon)}$; if $\mathrm{ext} = \min$ the two are exchanged. By Lemma \ref{lem:search.step} the invariant $lb \leq \mathrm{ext}_j\, R_{d+1}(w_j) \leq ub$ survives whichever update fires and whichever way an ambiguous child is classified.}

        \State Define $U_{min/max} = (U_{step})^{steps}$\grn{, applied with $\ket{lb} = \ket{0}$, $\ket{ub} = \ket{1}$} \Comment{Applies $U_{step}$ steps times starting with the lower bound 0 and upper bound 1.}
        \State Define $U_{d, \epsilon}$ to apply $U_{min/max}$ to $\ket{v}$ followed by $U_{pivot}$ and puts the fourth register of $U_{pivot}$ in the second (output) register given and all the other intermediate results in the third register as garbage.
        \State \Return $U_{d, \epsilon}$
    \end{algorithmic}
\end{algorithm}

\subsubsection{Query complexity of the deterministic step}\label{sec:minmax-algorithm-subsection-query-complexity}

Note that the unitary returned uses $\log_2(1/\grn{\epsilon_2})$ search iterations. Each search iteration has a query complexity of $O(\sqrt{deg}\grn{)} \cdot Q(U_{d+1\grn{,\epsilon/2}})$, but needs to be amplified with majority voting $\grn{O(\log(steps/\delta))}$ times, \grn{and the median-boosted $\bar U_{est}$ answers each query at a further cost of $r = O(\log (deg))$ copies,} so the total query complexity of the unitary is

\begin{equation} \label{eq:det.step.query.complexity}
Q(U_{d, \epsilon}) = O\Big(\sqrt{deg} \cdot \log(1/\grn{\epsilon_2}) \log \big(\tfrac{\log(1/\grn{\epsilon_2})}{\delta}\big) \grn{{}\cdot \log (deg)}\Big)\cdot Q(U_{d+1, \grn{\epsilon/2}})
\end{equation}

\grn{The children are asked for accuracy $\Theta(\epsilon)$, not $\Theta(\epsilon/deg)$. This keeps the deterministic layer's contribution to the exponent of $deg$ at $\frac12$ rather than $\frac32$, and it is possible because Lemma \ref{lem:search.step} brackets the true extremum directly instead of taking an extremum of mean-square estimates (cf.\ Remark \ref{rem:max.lemma.is.tight}).}

\subsubsection{Error budget of the deterministic step}

\begin{lemma}[Algorithm \ref{alg:min-max-binary-search} realises the deterministic step] \label{lem:minmax.satisfies.assumption}
    \grn{Let $type(d-1) = det$ and $G_d = \mathrm{ext} \in \{\max,\min\}$. Let $U_{d+1}$ be a unitary at layer $d+1$ whose outputs $R_{d+1}(v')$, $v' \in V(T)_d$, satisfy $\mathbb{E}[(R_{d+1}(v') - y_{v'})^2] \leq (\epsilon/8)^2$ for numbers $y_{v'} \in [0,1]$, and let $U_{d,\epsilon}$ be $MinMaxFindingUnitary(U_{d+1}, d, \frac{\epsilon}{4}, \epsilon_2 = \frac{\epsilon}{2}, \delta = \frac{\epsilon^2}{2}, \mathrm{ext})$ built on this $U_{d+1}$, with its output clipped to $[0,1]$. Then for every $v \in V(T)_{d-1}$
    \[ \mathbb{E}\big[(R_d(v,\epsilon) - \mathrm{ext}_j\, y_{w_j})^2\big] \leq \tfrac34 \epsilon^2 \leq \epsilon^2, \]
    and $U_{d,\epsilon}$ applies $U_{d+1}$ at most $q_d(\epsilon) = O\big(\sqrt{deg}\,\log(\tfrac1\epsilon)\,\log(\tfrac{\log(1/\epsilon)}{\epsilon})\,\log deg\big)$ times. Hence $MinMaxFindingUnitary$ with these parameters realises $ConstructUnitary\_d$ of Assumption \ref{ass:deterministic.step.unitary} with $\kappa_d = 8$ and this $q_d$.}
\end{lemma}

\begin{proof}
    \grn{By Lemma \ref{lem:median.bridge}, the median-boosted read $\bar U_{est}$ of child $j$ returns a value within $2 \cdot \frac{\epsilon}{8} = \frac{\epsilon}{4}$ of $y_{w_j}$ except with probability $e^{-r/8} \leq \frac{1}{\grn{C_0}\,deg}$. This is the hypothesis of Lemma \ref{lem:search.step} with slack $\frac{\epsilon}{4}$. Apply that lemma with $\delta' = \delta/steps$ at each of the $steps$ bracket updates (each update uses fresh copies of $\bar U_{est}$). The invariant $lb \leq \mathrm{ext}_j y_{w_j} \leq ub$ then holds throughout except on an event of probability at most $\delta$, whatever the comparison bits of the children within $\frac{\epsilon}{4}$ of the pivot. The bracket length satisfies $\lambda_{i+1} = \frac{\lambda_i}{2} + \frac{\epsilon}{4}$ with $\lambda_0 = 1$, so after $steps = \lceil\log_2(1/\epsilon_2)\rceil$ rounds it is at most $\frac{\epsilon}{2} + \epsilon_2 = \epsilon$, and on the good event the returned midpoint is within $\frac{\epsilon}{2}$ of $\mathrm{ext}_j y_{w_j}$. On the failure event both the output and the target lie in $[0,1]$. So the mean square error is at most $(\frac{\epsilon}{2})^2 + \delta \cdot 1 = \frac{\epsilon^2}{4} + \frac{\epsilon^2}{2} = \frac34\epsilon^2$. The cost is (\ref{eq:det.step.query.complexity}) with $\epsilon_2 = \frac{\epsilon}{2}$ and $\delta = \frac{\epsilon^2}{2}$.}
\end{proof}

}

\subsection{Chance layers} \label{sec:chance.layers}

\mrn{
In Lemma \ref{lem:non.det.assumption.is.satisfied}, we showed that there is always an algorithm $ConstructUnitary\_non\_det\_d$ that satisfies the assumptions of Assumption \ref{ass:non-deterministic.step.unitary} for a particular value $s^2 = \sup_{w \in V(T)_d} \left( \mathbb{E}[\Delta(w)^2] \right)$. In fact, the dependence of the constructed algorithm on the linear function $G_d(v, \cdot)$ is fully expressed in terms of the parameters $c_j(v)$. In the case of expectiminimax trees, the linear function $G_d$ is given by
\[ G_d(v, x_1, \dots, x_{deg}) = \frac{1}{deg} \sum_{i=1}^{deg} x_i. \]
In particular,
\begin{equation} \label{eq:Lambda.for.expectiminimax}
    \Lambda = \sup_{v} \sum_{i=1}^{deg} \vert c_j(v) \vert = 1 = deg L.
\end{equation}

We call the algorithm produced by $ConstructUnitary\_non\_det\_d$ in this particular case $MeanFindingUnitary$. 
}

\subsection{Assembling the layers}

\grn{Algorithm \ref{alg:build-unitary-estimator} assembles the two routines into the recursive family of unitaries of Algorithm \ref{alg:generalized.unitary.estimator.constructor}. A deterministic layer at accuracy $\epsilon$ calls $MinMaxFindingUnitary$ with slack $\frac{\epsilon}{4}$, resolution $\epsilon_2 = \frac{\epsilon}{2}$ and failure rate $\delta = \frac{\epsilon^2}{2}$. By Lemma \ref{lem:minmax.satisfies.assumption} its output then has mean square error at most $\frac34\epsilon^2$ and it calls its children at accuracy $\frac{\epsilon}{8}$. A chance layer calls $MeanFindingUnitary(U_{d+1,\epsilon}, s_n, \frac{\epsilon}{3(B_d+1)})$, $B_d+1$ many times. No global failure budget divided over the layers is needed, because what travels up the tree is a root mean square error and not a union-bounded success event.}

{\authorcolour{Maroon}
\begin{algorithm}
\caption{BuildUnitaryEstimator(d, $\epsilon$) \grn{(stochastic $2$-player games)}}\label{alg:build-unitary-estimator}

\begin{algorithmic}[1]
    \If{d = D + 1}
        \State Define $U_{d, \epsilon} \ket{v} \ket{0} \ket{0}:= \ket{v} \ket{val(v)} \ket{0}$
        \State \grn{\Return $U_{d,\epsilon}$}
    \EndIf

    \If{type(\grn{d-1}) = deterministic} \Comment{\grn{$U_{d,\epsilon}$ acts on vertices of depth $d-1$; the step to their children has type $type(d-1)$.}}

        \State $U_{d+1, \epsilon} \gets BuildUnitaryEstimate(d,\epsilon)$
        \State $U_{d,\epsilon} \gets MinmaxFindingUnitary(U_{d+1,\epsilon}, d, \frac{\epsilon}{4}, \epsilon_2 = \frac{\epsilon}{2}, \delta = \frac{\epsilon^2}{2}, \mathrm{ext} = \tau(d-1))$ \Comment{This algorithm is our $ConstructUnitary\_d(U_{d+1},\epsilon)$.}


    \EndIf

    \If{type(\grn{d-1}) = non-deterministic}

        \State $B_{\grn{d}} \gets \grn{\max\Big(0,\; \grn{\Big\lceil} 2 \log_2\Big( \frac{ 2 \sqrt{2}\, deg L \kappa_d }{\epsilon} \Big) \grn{\Big\rceil}\Big)}$
        \For{$0 \leq n \leq B_{\grn{d}}$}
            \State $U_{d+1,n} \gets BuildUnitaryEstimator(d+1, 2^{-\frac{n}{2}})$

            \State $U_{\Delta,n} \gets$ the unitary that applies $U_{\grn{d+1},n}$ and $U_{\grn{d+1},n-1}$ on separate workspaces, unitarily subtracts the result of $U_{\grn{d+1}, n-1}$ from the result of $U_{\grn{d+1},n}$ and puts this in the second register, putting all garbage and intermediate outputs in the last register. \Comment{If $n=0$, we simply set $U_{\Delta,0} \gets U_{\grn{d+1},0}$.}

            \State \grn{$s_n \gets 3 \cdot 2^{-n/2}$}

            \State $A_{d,n} \gets MeanFindingUnitary\big(U_{\Delta,n}, \grn{s_n,}\ \grn{\frac{\epsilon}{3(B_d+1) \kappa_d}}\big)$

        \EndFor
        \State $U_{d, \epsilon}$ is defined as the unitary that takes a $\ket{v}, v \in V(T)_{\grn{d-1}}$ as input and for every $n$, applies $A_{d,n}$ to $\ket{v}$ on a separate workspace\grn{,}  coherently computes their sum in the output register\grn{, and clips it to $[0,1]$}. 

        \State \Return $U_{d, \epsilon}$
        
    \EndIf
    \State \Return $U_{d, \epsilon}$
\end{algorithmic}
\end{algorithm}

}

\subsection{The main bound} \label{sec:analysis.unitary}

\grn{In this section $\epsilon$ is the accuracy requested at the root. A subroutine further down is called at some $\delta' \geq \epsilon_{\min}:= \vio{\epsilon_{\min,D+1} = }\grn{\epsilon/32^{m}}$\vio{, the smallest accuracy floor of Theorem \ref{thm:unitary.general} for the root step $d = 1$ with $\kappa_j = 8$ and $\Lambda_j = 1$}: a deterministic layer at accuracy $\delta$ calls its children at $\frac{\delta}{8}$, a chance layer at $2^{-n/2} \geq 2^{-B_d/2} \grn{\geq \frac{\delta}{4}}$, and these compose over \grn{$m$ layer pairs}. We evaluate every overhead at $\epsilon_{\min}$, which puts $D$ inside the logarithms: $\log(1/\epsilon_{\min}) = \grn{m\log 32} + \log(1/\epsilon) = O(D + \log(1/\epsilon))$. \grn{The factor $K^{D}$ of Theorem \ref{thm:headline} is a different constant: it collects the constants of the cost recursion.} A deterministic layer multiplies the $\frac{1}{\epsilon}$ that passes through it by $8$ and by the constant of the search, and a chance layer by the $16c_Q$ of (\ref{eq:unitary.query.complexity}). Remark \ref{rem:regime} takes this into account.}

\begin{theorem}[Main bound] \label{thm:headline}
    \grn{Let $T$ be a stochastic $2$-player game tree, that is, an expectiminimax tree (Definition \ref{def:emm.tree}) with $\tau$ alternating, of even depth $D = 2m$ with degree at most $deg$, leaf values in an interval of length $N$, and $type(D-1) = det$. Then $U_{1,\epsilon}$ of Algorithm \ref{alg:build-unitary-estimator} returns an estimate of $\gamma_1(root)$ with RMSE at most $\epsilon$ using}
    \[ \grn{Q(U_{1,\epsilon}) \; \leq \; O\!\left( K^{D}\, deg^{\frac{D}{4}} \, \frac{N}{\epsilon} \, D^{4D} \log\left( \frac{N}{\epsilon} \right)^{4D} \right).} \]
    \grn{queries to the leaf oracle, where $K$ is an absolute constant\grn{, at most the larger of the two per-layer constants of Theorem \ref{thm:unitary.general}; it is determined by $c_Q$ and by the constants of Lemma \ref{lem:search.step} and Algorithm \ref{alg:min-max-binary-search}}. With $D = 2m$ this is $\widetilde O(deg^{m/2}\epsilon^{-1})$ against the classical $\widetilde O(deg^{m}\epsilon^{-2})$ of (\ref{eq:classical.baseline}), a quadratic improvement in both parameters, where $\widetilde O$ hides factors polylogarithmic in $deg$, $N$ and $\frac{1}{\epsilon}$ and exponential in $D$.}
\end{theorem}

\begin{proof}
    \grn{Algorithm \ref{alg:build-unitary-estimator} is a special case of Algorithm \ref{alg:generalized.unitary.estimator.constructor} with $\kappa_d = 8$ and, by Lemma \ref{lem:minmax.satisfies.assumption},
    \[ q_d(\delta) = O\big(\sqrt{deg}\,\log(\tfrac1\delta)\log(\tfrac{\log(1/\delta)}{\delta})\log deg\big) \]
    at the $m$ deterministic steps, and $deg L = 1$ at the $m$ non-deterministic steps. In particular, we see that $\beta = 3$. Theorem \ref{thm:unitary.general} gives the RMSE bound and
    \begin{equation*}
        \begin{split}
            Q(U_{1,\epsilon}) & \;\leq\; O\left( \frac{1}{\epsilon} deg^{\alpha m} K'^c \kappa^c D^{4D} \log\left( \frac{\kappa}{\epsilon} \right)^{4D} \right)\\
        \end{split}
    \end{equation*}
    By Lemma \ref{lem:minmax.satisfies.assumption}, we have $\alpha = \frac{1}{2}$. We thus obtain
    \[ Q(U_{1, \epsilon}) \leq O\left( \frac{1}{\epsilon} deg^{\frac{D-d}{4}} K'^{\frac{D-d}{2}} \kappa^{\frac{D-d}{2}} D^{4D} \log\left( \frac{8}{\epsilon} \right)^{4D} \right). \]}
    We can pack $\kappa$ into the constant $K := \kappa K'$. After the rescaling of Definition \ref{def:no.blowup}, which replaces $\epsilon$ by $\epsilon/N$, we have $\log(1/\epsilon) = O(\log(N/\epsilon))$. This produces the bound on $Q(U_{1,\epsilon})$ stated in the Theorem. Lemma \ref{lem:coherent.composition} justifies composing the layers inside the mean estimator.
\end{proof}

\begin{corollary}[Constant depth] \label{cor:constant.depth}
    \grn{If we treat the depth $D$ as a constant, as the analyses of nested expectation estimation in \cite{sunOptimalQuantumSpeedups2026,blanchetNonlinearQuantumMonteCarlo2025} do, Theorem \ref{thm:headline} reads
    \[ Q(U_{1,\epsilon}) \;=\; \widetilde O\!\left( deg^{m/2}\, \epsilon^{-1} \right) , \]
    where $\widetilde O$ hides factors polylogarithmic in $deg$, $N$ and $\frac1\epsilon$. By Proposition \ref{prop:composed.lower.bound} this is optimal up to a factor $\sqrt{deg}$ and the polylogarithm whenever $\frac{1}{4\,deg} \leq \epsilon \leq \frac18$ \grn{and $m \geq 2$}. Without that convention, the bound improves on the classical baseline (\ref{eq:classical.baseline}) in $\epsilon$ for every tree, and in $deg$ for trees with $deg \gtrsim K^4\Lambda^{12}$ (Remark \ref{rem:regime}).}
\end{corollary}

\begin{corollary}[Unconditional bound with evaluated deterministic layers] \label{cor:unconditional}
    \grn{Suppose the deterministic layers of Algorithm \ref{alg:build-unitary-estimator} evaluate all $deg$ children, each boosted by a median over $r = O(\log(deg/\epsilon))$ copies as in Lemma \ref{lem:det.step.no.deg}, and take the extremum with reversible arithmetic, in place of Algorithm \ref{alg:min-max-binary-search}. Then}
    \[ \grn{Q(U_{1,\epsilon}) \; = \; O\!\left( K^{D}\, deg^{\frac{D}{2}} \, \frac{N}{\epsilon} \, D^{4D} \log\left(\frac{N}{\epsilon}\right)^{4D} \right) \; = \; \widetilde O\!\left( deg^{m} \frac{1}{\epsilon} D^{4D} \right) \qquad \text{for } D = 2m.} \]
    \grn{Against the classical $\widetilde O(deg^{m}\epsilon^{-2})$ of (\ref{eq:classical.baseline}) this is a quadratic improvement in $\epsilon$ and none in $deg$.}
\end{corollary}

\begin{proof}
    \grn{This is instance (b) of \vio{Example} \ref{rem:det.instances}, where $\alpha = 1$. Plugging this into the formula of Theorem \ref{thm:unitary.general} and using the same values as in Theorem \ref{thm:headline} otherwise, we get the stated bound.}
\end{proof}

\grn{The corollary separates the two speedups. The improvement in $\epsilon$ rests only on the telescoping of the level differences and does not care how the deterministic layers are done. The improvement in $deg$ is what Algorithm \ref{alg:min-max-binary-search} buys, and it is the part that needs Lemma \ref{lem:search.step}.}

\section{Optimality and the regime of advantage} \label{sec:optimality}

\grn{Prior to this work no lower bound was known for expectiminimax trees in which both kinds of layer are active as far as we are aware. What we have are two degenerate families. Each is itself an expectiminimax tree, so a lower bound on either transfers to the general problem. Collapse the chance layers, that is, make all children of every chance node equal so that each expectation is the identity. What remains is an $m$-layer minimax tree on $deg^m$ effective leaves, and one query to it answers every leaf query in the corresponding block. Its boolean special case is an AND-OR formula and needs $\Omega(\sqrt{deg^m}) = \Omega(deg^{m/2})$ queries \cite{barnumLowerBoundQuantum2004}. The reduction from Min-Max to And-Or trees is the one \textcite{cleveQuantumAlgorithmsEvaluating2019} use for the same purpose, and \textcite{ambainisAnyAndOrFormula2010} give the matching upper bound. This is the $deg^{m/2}$ of \cite{cleveQuantumAlgorithmsEvaluating2019}, and Theorem \ref{thm:headline} matches it. Collapse the deterministic layers instead, making all children of every max and min node equal. What remains is the nested expectation over root-to-leaf paths, that is, mean estimation of a bounded random variable, which needs $\widetilde\Omega(\epsilon^{-1})$ \cite{nayakQuantumQueryComplexity1999,hamoudiQuantumAlgorithmsMonte}. Together these give}
\[ \grn{Q_\epsilon(T) \; = \; \widetilde\Omega\!\left( \max\left( deg^{m/2}, \; \epsilon^{-1} \right) \right),} \]
\grn{the maximum of the two and not their product, while the upper bound above is the product. A matching $\Omega(deg^{m/2}\epsilon^{-1})$ needs an instance in which both layer types are active, which neither degenerate family provides. Proposition \ref{prop:composed.lower.bound} below gives such an instance through the composition theorem for the general adversary bound. It proves the product up to a factor $deg^{m_2(\epsilon)/2}$, where $m_2(\epsilon)$ is the number of layer pairs a chance block needs to accrue an error of at least $\epsilon$; this is a single factor $\sqrt{deg}$ whenever $\epsilon \geq \frac{1}{4\,deg}$. So the algorithm is optimal in each parameter on its own with the other held fixed, and jointly up to $\sqrt{deg}$ and the polylogarithm in that regime.}

\grn{The same two collapses give the classical lower bounds recorded with (\ref{eq:classical.baseline}): $\Omega(\rho_{deg}^{\,m})$ with $\rho_{deg} \geq deg/2$ from \cite[Thm.~5.5]{saksProbabilisticBooleanDecision1986} and \cite{santhaMonteCarloBoolean1995}, and $\Omega(\epsilon^{-2})$ from classical mean estimation. So the quantum bound beats every classical algorithm in each parameter.}

\grn{For a quantum lower bound with both kinds of layer active we compose the two hard instances. \mrn{To execute this, we use} the general adversary bound $\mathrm{Adv}^{\pm}$ of \cite{hoyerNegativeWeightsMake2007}. It characterises bounded-error quantum query complexity for every function with finite domain and range, partial functions included \cite[Thm.~1.1]{leeQuantumQueryComplexityState2011}, and it is multiplicative under composition when the outer function is total and the inner function has boolean output \cite[Lemma~5.2]{leeQuantumQueryComplexityState2011}, where the inner function may be partial. In the boolean case this is \cite[Thm.~13]{hoyerNegativeWeightsMake2007} and \cite[Thm.~2.7]{reichardtSpanPrograms2009}.}

\begin{proposition}[A lower bound with both layer types active] \label{prop:composed.lower.bound}
    \grn{Let $T$ be the alternating tree of depth $D = 2m$, degree $deg \geq 2$ and $type(D-1) = det$, with leaf values in $\{0,1\}$, and for $0 < \epsilon \leq \frac18$ put $m_2(\epsilon) := \lceil \log_{deg}(1/(4\epsilon)) \rceil$, the number of layer pairs a chance block needs to carry a gap of $2\epsilon$ in its mean. If $m_2(\epsilon) \leq m-1$, then every quantum algorithm that estimates the value of $T$ with root mean square error at most $\epsilon$ makes
    \[ \Omega\!\left( deg^{\frac{m - m_2(\epsilon)}{2}}\, \epsilon^{-1} \right) \]
    queries to the leaf oracle. In particular the bound is $\Omega(deg^{\frac{m-1}{2}}\epsilon^{-1})$ for $\frac{1}{4\,deg} \leq \epsilon \leq \frac18$, and it is at least $\Omega(deg^{\frac{m-1}{2}}\epsilon^{-1/2})$ whenever $\frac{1}{4\,deg^{\,m-1}} \leq \epsilon \leq \frac18$.}
\end{proposition}

\grn{The proof is in Appendix \ref{app:proof.lower}.}

\grn{Together with the two degenerate families this gives
\[ Q_\epsilon(T) \;=\; \Omega\!\left( \max\!\left( deg^{m/2},\; \epsilon^{-1},\; deg^{\frac{m - m_2(\epsilon)}{2}}\epsilon^{-1} \right) \right) . \]
The upper bound of Theorem \ref{thm:headline} exceeds the third term by $deg^{m_2(\epsilon)/2}$, the square root of the number of layer pairs a gap block occupies. For $\epsilon \geq \frac{1}{4\,deg}$, that is, whenever the requested accuracy is no finer than one reciprocal branching factor, this is a single factor $\sqrt{deg}$ over the whole depth. In that regime, which is also the regime of Remark \ref{rem:regime}, the product $deg^{m/2}\epsilon^{-1}$ is the true complexity up to $\sqrt{deg}$ and the polylogarithm. The deterministic and the chance costs compound, and this is a property of the problem, not of the algorithm. For finer accuracies the gap grows by $\sqrt{deg}$ for each further factor $deg$ in $\epsilon^{-1}$ and never exceeds $\sqrt{deg/(4\epsilon)}$. The source of the loss is cleear: a chance node with $deg$ children can carry a gap of at most $\frac{1}{2\,deg}$ in its mean, so a finer gap has to be spread over $m_2$ chance layers. In an alternating tree, these arrive interleaved with $m_2$ deterministic layers that the hard instance has to switch off, and switching them off costs $deg^{m_2/2}$ of the outer bound. This is a limitation of the hard instance, not of alternation as a model. An instance in which the deterministic layers inside a block stay active, or an algorithm that spends precision only where the bracket is tight, would settle which side the last factor belongs to. The telescoping in the proof of Theorem \ref{thm:unitary.general} is a case in this same setting where per-layer costs that appear to multiply do not.}

\begin{remark} \label{rem:regime}
    \grn{The per-layer overheads are raised to the power $D$, so the phrase ``logarithmic factors'' carries an exponent that grows with the depth, and we have to state the regime. With $M = D + \log(deg\,N/\epsilon)$ as in Theorem \ref{thm:headline}, the bound beats the classical baseline (\ref{eq:classical.baseline}) in $deg$ only when $K^{4}M^{12} \lesssim deg$, and beats the trivial algorithm that reads all $deg^{D}$ leaves when $(KM^{3})^{4/3} \lesssim deg$. Since $M \geq D$, both are statements about trees that are very wide compared to their depth. \grn{Below the first threshold, evaluating all children as in Corollary \ref{cor:unconditional} gives the better bound: it matches the classical exponent in $deg$ and keeps the speedup in $\epsilon$.} The improvement in $\epsilon$, from $\epsilon^{-2}$ to $\epsilon^{-1}$, holds in every regime. The classical multilevel scheme behind (\ref{eq:classical.baseline}) carries its own constant per layer, roughly the square of ours because an accuracy split costs $\epsilon^{-2}$ there, so the constants do not reverse the comparison with the baseline. They do place both thresholds at very large $deg$, and Theorem \ref{thm:headline} is an asymptotic statement. \grn{As the proofs stand, the deterministic-layer constant is of order $10^4$ (the $8$ of $\kappa_d$, the $18$ of the Hoeffding majority, the four sets and the doubling schedule of Appendix \ref{app:proof.search}, and the $8$ of the median boost) and the chance-layer constant is $16c_Q$, with $c_Q$ the constant of \cite[Cor.~3.2]{sunOptimalQuantumSpeedups2026}. None of these is optimised; a fixed-point search \cite{yoderFixedPointQuantumSearch2014} in place of the doubling schedule would remove most of the first.} The classical lower bound recorded with (\ref{eq:classical.baseline}) holds for all $deg$ and $\epsilon$, so the separation in the exponents does not depend on the constants.}
\end{remark}

\section{Open problems and outlook} \label{sec:ideas}

\grn{To conclude, we list four open problems, although \mrn{there likely are} be more.}

\grn{\emph{Adaptive subroutines.} Lemma \ref{lem:coherent.composition} justifies composing the layers inside the mean estimator, but its condition (ii), one circuit with all parameters independent of the vertex, rules out every adaptive or variable-time subroutine. \mrn{This restriction} matters since  D\"urr--H\o yer \cite{durrQuantumAlgorithmFinding1999} is not applicable; although it beats Algorithm \ref{alg:min-max-binary-search} by the logarithmic factors in $q_d$. Whether a variable-time amplitude estimation in the style of \cite{ambainisQuantumAlgorithmTree2017} can be made to compose here, and so recover those factors, is open.}

\grn{\emph{A matching lower bound.} Section \ref{sec:optimality} gives $\Omega(\max(deg^{m/2}, \epsilon^{-1}, deg^{(m-m_2(\epsilon))/2}\epsilon^{-1}))$, the last term from a composed instance in which both kinds of layer are active (Proposition \ref{prop:composed.lower.bound}), while the upper bound of Theorem \ref{thm:headline} is the product $\widetilde O(deg^{m/2}\epsilon^{-1})$. For $\epsilon \geq \frac{1}{4\,deg}$ the two differ by a single factor $\sqrt{deg}$ and the polylogarithm. For finer accuracies the remaining factor $deg^{m_2(\epsilon)/2}$ is the price the composed instance pays for the deterministic layers it has to switch off inside each gap block. Closing the gap needs either a hard instance in which those layers stay active, or an algorithm that spends precision only where the bracket is tight.}

\grn{\emph{The classical product.} Section \ref{sec:complexity.measure} records the classical lower bound $\Omega(\max(\rho_{deg}^{\,m}, \epsilon^{-2}))$, which makes the baseline (\ref{eq:classical.baseline}) tight in each parameter on its own. Whether the classical product $\Omega(deg^m\epsilon^{-2})$ holds is open, as is its quantum analogue. A pruned classical algorithm for approximately evaluated expectiminimax trees that beats the baseline would be of interest in its own right.}

\emph{\grn{Bounded quantum memory.}} It might be worth proving that a speedup survives in the setting where the quantum computer has a size that is a constant fraction of the problem size, as in \cite{ambainisQuantumAlgorithmTree2017}, both for minimax and for expectiminimax. \grn{This is the most practical of the four. The algorithm as stated holds the whole recursion coherently, and a version that keeps only a constant fraction of the tree in superposition would be far closer to what a device can do.}

Showing that expectiminimax-like trees can be solved coherently may lead to other algorithms. The algorithm is not practical because of its prefactors. However, if it is combined with an approximate algorithm that uses heuristic values after some depth, the speedup may become useful. For realistic games, the prefactor is unlikely to give a speedup. But because the method is general, it may help for more complex stochastic games. Now that we know that the asymptotic speedup can hold, a different approach could reduce the prefactor or polylog factors. 

\printbibliography

\appendix

\section{Proof of Theorem \ref{thm:unitary.general}} \label{sec:proof.2.of.main.theorem}

\begin{proof}

Let $m$ denote the number of deterministic steps from $d$ to $D$ and $c$ the number of non-deterministic steps from $d$ to $D$. To show that Algorithm \ref{alg:generalized.unitary.estimator.constructor} has the right bounds on root mean square error and complexity, we apply induction. We start with the root mean square error. For $d = D+1$, the output is exactly correct, so the error is zero. For the induction step, we have to distinguish between the deterministic and the non-deterministic step.

Let $R_{d}(v, \epsilon)$ be the random variable obtained from measuring the output register of $U_{d, \epsilon} \ket{v} \ket{0} \ket{0}$, where $v \in V(T)_{d-1}$. As we will see, the proper induction step will go through two recursion steps, one deterministic and one non-deterministic one. Suppose by induction that, for every $\delta > 0$,
\[ \mathbb{E}[ (\gamma_{d+2}(w) - R_{d+2}(w,\delta))^2 ] < \delta^2. \]
We will do the case where $type(d-1) = det$. If $type(d-1) = non-det$, we can do one non-deterministic recursion step following the calculations below and then apply the double recursion we are about to do. We may thus assume without loss of generality that $type(d-1) = det$.\\

For any $v' \in V(T)_{d}$, let $A_{d+1,n}(v')$ denote the random variable obtained from measuring the output register of $A_{d+1,n} \ket{v'}\ket{0}\ket{0}$, where $A_{d+1,n}$ are the unitaries that are computed during the first recursion when applying $BuildUnitaryEstimator(d,\epsilon)$. Note that $R_{d+1}(v',\frac{\epsilon}{\kappa_d}) = \sum_{n=0}^{B_{d+1}} A_{d+1,n}(v')$ by definition of $U_{d,\epsilon}$. Similarly, let $\Delta_{d+2}(w,n)$ denote the random variable obtained from $U_{\Delta,n} \ket{w} \ket{0} \ket{0}$. We need to show that
\[ \mathbb{E}[(R_d(v,\epsilon) - \gamma_d(v))^2] \leq \epsilon^2. \]
By Assumption \ref{ass:deterministic.step.unitary} Point 3, this follows, if
\[ \mathbb{E}[(R_{d+1}(v',\frac{\epsilon}{\kappa_d}) - \gamma_{d+1}(v'))^2] \leq \frac{\epsilon^2}{\kappa_d^2} \qquad \text{and} \qquad \vert R_{d+1}(v,\frac{\epsilon}{\kappa_d}) - \gamma_{d+1}(v') \vert \leq 1 \] 
for all $v' \in children(v)$. (Note that the unitary $U_{d+1}$ used to produce $U_{d,\epsilon}$ works with error bound $\frac{\epsilon}{\kappa_d}$, hence the $\frac{\epsilon}{\kappa_d}$-error bound in $R_{d+1}(v', \frac{\epsilon}{\kappa_d})$.) The second of these two inequalities is always satisfied because of our chosen normalization and because we clip the $R_{d+1}(v',\epsilon)$ into $[0,1]$. We are thus left to estimate the left-hand-side of the first expression. For any $v' \in children(v)$, we have that

We now estimate
\begin{equation} \label{eq:unitary.RMSE.bound}
    \begin{split}
        & \leq \mathbb{E} \left[ \left( R_{d+1}(v',\frac{\epsilon}{\kappa_d}) - \gamma_{d+1}(v') \right)^2 \right]\\
        & \leq 2 \mathbb{E}\left[ \left( R_{d+1}(v',\frac{\epsilon}{\kappa_d}) - \sum_{n=0}^{B_{d+1}} G_{d+1}(v',\Delta_{d+2}(w)) \right)^2 \right]\\
        & \qquad + 2 \mathbb{E}\left[ \left( \sum_{n=0}^{B_{d+1}} G_{d+1}(v',\Delta_{d+2}(w)) - \gamma_{d+1}(v') \right)^2 \right]\\
        & = 2 \mathbb{E}\left[ \left( \sum_{n=0}^{B_{d+1}} (A_{d+1,n}(v') - G_{d+1}(v',\Delta_{d+2}(w))) \right)^2 \right]\\
        & \qquad + 2 \mathbb{E}\left[ \left( \sum_{n=0}^{B_{d+1}} G_{d+1}(v',\Delta_{d+2}(w)) - G_{d+1}(v', \gamma_{d+2}(w)) \right)^2 \right]\\
    \end{split}
\end{equation}
where we inserted a term, applied Lemma \ref{lem:CS.application}, and used the definitions of $A_{d+1,n}(v')$ and $\gamma_{d+1}(v')$. We now bound the first term in this expression by using Lemma \ref{lem:CS.application} and the RMSE bound of $ConstructUnitary\_non\_det\_d+1(U_{\Delta,n}, s_n, \frac{\epsilon}{3(B_d + 1) \kappa_d})$:
\begin{equation} \label{eq:unitary.first.summand}
    \begin{split}
        \mathbb{E} & \left[ \left( \sum_{n=0}^{B_{d+1}} (A_{d+1,n}(v') - G_{d+1}(v',\Delta_{d+2}(w))) \right)^2 \right]\\
        & \leq (B_{d+1} + 1) \sum_{n=0}^{B_{d+1} + 1} \mathbb{E} \left[ \left( A_{d+1,n}(v') - G_{d+1}(v',\Delta_{d+2}(w)) \right)^2 \right]\\
        & \leq (B_{d+1} + 1)^2 \frac{\epsilon^2}{9 deg^2 L^2 (B_{d+1}+1)^2 \kappa_{d+1}^2}\\
        & \leq \frac{\epsilon^2}{9 deg^2 L^2 \kappa_{d+1}^2}.
    \end{split}
\end{equation}

To bound the second term, we first note that $G_{d+1}(v', \gamma_{d+2}(w))$ is a constant if $w$ is fixed. Furthermore, we highlight that the definiton of $U_{\Delta}$ gives us a telescope sum, where
\[ \sum_{n=0}^{B_{d+1}} \Delta_{d+2}(w,n) = R_{d+2}\left(w, 2^{- \frac{B_{d+1}}{2}} \right). \]
By linearity of $G_{d+1}$, this implies that
\begin{equation*}
    \begin{split}
        \sum_{n=0}^{B_{d+1}} G_{d+1}(v', \Delta_{d+1}(w,n)) & = G_{d+1}\left(v', \sum_{n=0}^{B_{d+1}} \Delta_{d+1}(w,n)\right)\\
        & = G_{d+1}\left(v', R_{d+2}\left(w, 2^{- \frac{B_{d+1}}{2}} \right) \right)
    \end{split}
\end{equation*}

Recalling that $B_{d+1} = 2 \log_2\left( \frac{2 \sqrt{2} deg L \kappa_{d+1}}{\epsilon} \right)$, we can bound the second term as follows:
\begin{equation} \label{eq:unitary.second.summand}
    \begin{split}
        \mathbb{E} & \left[ \left( \sum_{n=0}^{B_{d+1}} G_{d+1}(v', \Delta_{d+2}(w,n)) - G_{d+1}(v', \gamma_{d+2}(w)) \right)^2 \right]\\
        & = \mathbb{E} \left[ \left( G_{d+1}(v', R_{d+2}(w, 2^{-\frac{B_{d+1}}{2}}) - G_{d+1}(v', \gamma_{d+2}(w)) \right)^2 \right]\\
        & \leq L^2 \mathbb{E} \left[ \left( \sum_{w \in children(v')} \vert R_{d+2}(w, 2^{-\frac{B_{d+1}}{2}}) - \gamma_{d+2}(w) \vert \right)^2 \right]\\
        & \leq deg L^2 \sum_{w \in children(v')} \mathbb{E} \left[ \left( R_{d+2}(w, 2^{-\frac{B_{d+1}}{2}}) - \gamma_{d+2}(w) \right)^2 \right]\\
        & \leq deg L^2 \sum_{w \in children(v')} 2^{-B_{d+1}}\\
        & \leq deg^2 L^2 \frac{\epsilon^2}{8 deg^2 L^2 \kappa_{d+1}^2}\\
        & = \frac{\epsilon^2}{8 \kappa_{d+1}^2}.
    \end{split}
\end{equation}

Plugging estimates (\ref{eq:unitary.first.summand}) and (\ref{eq:unitary.second.summand}) into (\ref{eq:unitary.RMSE.bound}), we obtain
\[ (\ref{eq:unitary.RMSE.bound}) \leq 2 \frac{1}{9 deg^2 L^2 \kappa_{d+1}^2} \epsilon^2 + 2 \frac{1}{8 \kappa_{d+1}^2} \epsilon^2 < \frac{\epsilon^2}{\kappa_{d+1}^2}. \]
By Assumption \ref{ass:deterministic.step.unitary} Point 3, this implies that
\[ \mathbb{E}[(R_{d}(v,\epsilon) - \gamma_d(v))^2] \leq \epsilon^2. \]
We conclude that, from the induction assumption, we get the desired RMSE bound.

If $type(d-1) = non-det$, we apply inequality (\ref{eq:unitary.RMSE.bound}) directly and obtain the required bound after just one step of the recursion. (Recall that $\kappa_d \geq 1$.) This finishes the proof of the RMSE bound.\\

To estimate the complexity, we start by estimating the second momentum of $\Delta_{d+1}(w,n)$. For $n > 0$, we use the established RMSE bounds to estimate

\begin{equation} \label{eq:unitary.second.momentum}
    \begin{split}
        \mathbb{E} & \left[ \Delta_{d+2}(w,n)^2 \right]\\
        & = \mathbb{E} \left[ (R_{d+2}(w, 2^{-\frac{n}{2}}) - R_{d+2}(w, 2^{-\frac{n-1}{2}}))^2 \right]\\
        & \leq 2 \mathbb{E} \left[ \left(R_{d+2}(w, 2^{-\frac{n}{2}}) - \gamma_{d+2}(w) \right)^2 \right]\\
        & \qquad + 2 \mathbb{E} \left[ \left( \gamma_{d+2}(w) - R_{d+2}(w, 2^{-\frac{n-1}{2}}) \right)^2 \right]\\
        & \leq 2 (2^{-n} + 2^{-n+1})\\
        & = 6 \frac{1}{2^n}
    \end{split}
\end{equation}

For $n = 0$, we do a similar chain of inequalities to obtain
\[ \mathbb{E}[\Delta_{d+2}(w,0)] \leq 2 \cdot 2^0 + 2 \Gamma^2, \]
where $\Gamma := \sup_{w \in V(T)_{d+1}} \gamma_{d+1}(w)$. The extra constant $\Gamma$ in the case $n=0$ is the reason why we assume that there is no linear blowup. With no linear blowup, we can control this constant precisely in terms of the range of the values of the leaves. In particular, we can easily rescale and translate the initial values to make this extra constant be equal to one. Thus, $\Gamma$ will not affect the complexity estimate.\\

We now use this bound to compute the cost of our algorithm. We set $s_n = 3 \frac{1}{2^{\frac{n}{2}}}$ so that $\mathbb{E}[\Delta_{d+1}(w,n)] \leq s^2$ for every $w$. For $n=0$, we set
\[ s_0 = \sqrt{ 2 + 2 \Gamma^2} \leq 2 + 2 \Gamma \leq O(1 + \Gamma). \]
Since $s$ for the case $n = 0$ can be packed into an $O(1 + \Gamma)$, we will not treat this summand specially in the sums below.

The query complexity of the lowest level of the tree is straight-forward and serves as our induction start. We perform induction over two iterations of the recursion and assume that $type(d-1) = det$ just like for the RMSE bound. As before, the case where $type(d-1) = non-det$ has no additional obstacles. 

Our induction assumption consists of two statements: First, we assume that the query-complexity $Q_{d+2}(U_{d+2,\epsilon})$ to compute $R_{d+2}(w, \epsilon)$, where $w \in V(T_{d+1})$, satisfies equation (\ref{eq:general.theorem}). Second, we assume that there is a non-increasing function $\tilde{Q}_{d+2}(\epsilon)$ for $\epsilon > 0$ such that
\[ Q_{d+2}(U_{d+2,\epsilon}) \leq \frac{1}{\epsilon} \tilde{Q}_{d+2}(\epsilon). \]
As part of the induction, we will identify the precise structure of $\tilde{Q}_{d+2}(\epsilon)$.

The query complexity contributions of $ConstructUnitary\_d$ and $ConstructUnitary\_non\_det\_d+1$ are described in Assumptions \ref{ass:deterministic.step.unitary} and \ref{ass:non-deterministic.step.unitary}, point 4 respectively. The other complexity contribution come from the For-loop over $0 \leq n \leq B_{d+1}$ and the recursive use of $BuildUnitaryEstimator$ in line 13 of the algorithm. Thus, we obtain
\begin{equation} \label{eq:unitary.query.complexity}
    \begin{split}
        Q_{d}(U_{d,\epsilon}) & \leq q_d(\epsilon) Q_{d+1}\left(\frac{\epsilon}{\kappa_d}\right)\\
        & = q_d(\epsilon) \sum_{i=0}^{B_{d+1}} c_Q \frac{\Lambda_{d+1} s_n 3 (B_{d+1} + 1) \kappa_{d+1}}{\epsilon} \log\left( \frac{\Lambda_{d+1} s_n 3 (B_{d+1} + 1) \kappa_{d+1}}{\epsilon} \right)(Q_{d+2}(2^{-\frac{n}{2}}) + Q_{d+2}(2^{-\frac{n-1}{2}}))\\
        & \leq q_d(\epsilon) c_Q \frac{\Lambda_{d+1} 3 (B_{d+1} + 1) \kappa_{d+1}}{\epsilon} \sum_{i=0}^{B_{d+1}} s_n \log\left( \frac{9 \Lambda_{d+1} (B_{d+1} + 1) \kappa_{d+1}}{\epsilon} \right) 2^{\frac{n}{2}} (1 + 2^{-\frac{1}{2}}) \tilde{Q}_{d+2}(2^{-\frac{B_{d+1}}{2}})\\
        & \leq \frac{1}{\epsilon} q_d(\epsilon) 3 c_Q (B_{d+1} + 1) \Lambda_{d+1} \kappa_{d+1} \sum_{i=0}^{B_{d+1}} 3 \frac{3}{2} \log\left( \frac{9 \Lambda_{d+1} (B_{d+1} + 1) \kappa_{d+1}}{\epsilon} \right) \tilde{Q}_{d+2}\left(2^{-\frac{B_{d+1}}{2}}\right)\\
        & = \frac{1}{\epsilon} q_d(\epsilon) \frac{27}{2} c_Q (B_{d+1} + 1)^2 \Lambda_{d+1}\kappa_{d+1} \log\left(\frac{9 \Lambda_{d+1} (B_{d+1} + 1)\kappa_{d+1}}{\epsilon} \right) \tilde{Q}_{d+2}\left(2^{-\frac{B_{d+1}}{2}}\right).
    \end{split}
\end{equation}

We can thus write $Q_d(U_{d,\epsilon}) = \frac{1}{\epsilon} \tilde{Q}_{d}(\epsilon)$ with
\begin{equation*}
    \tilde{Q}_d(\epsilon) = q_d(\epsilon) \frac{27}{2} c_Q (B_{d+1} + 1)^2 \Lambda_{d+1} \kappa_{d+1} \log \left( \frac{9 \Lambda_{d+1} (B_{d+1} + 1) \kappa_{d+1}}{\epsilon} \right) \tilde{Q}_{d+2}\left( 2^{-\frac{B_{d+1}}{2}} \right).
\end{equation*}

We use this equation to bound $\tilde{Q}_d(\epsilon)$ from above up to some constants. Set
\[ K := \frac{27}{2} c_Q, \qquad \kappa := \max(\kappa_d), \qquad \bar{\Lambda} := \max(\Lambda_d) \leq deg L \]
We can start with $\tilde{Q}_{D+1}(\epsilon) = 1$. Every double-recursion -- going through one deterministic and one non-deterministic layer -- accrues a multiplicative factor $q_d(\epsilon) \frac{27}{2} c_Q (B_{d+1} + 1)^2 \Lambda_{d+1} \kappa_{d+1}$ for the appropriate $d$. Note that $B_{d+1} = O\left(\log\left(\frac{deg L \kappa}{\epsilon}\right)\right)$ depends on the $\epsilon$ used. Furthermore, the input into $\tilde{Q}_d$ and $\tilde{Q}_{d+2}$ differs by a factor $2 \sqrt{2} deg L \kappa$. This factor comes in once per double-recursion. As a consequence, we see that $B_{j} = O\left(\log\left( \frac{(deg L \kappa)^{\frac{D-j}{2}}}{\epsilon} \right) \right)$. Plugging this into $\tilde{Q}_d(\epsilon)$, we get the following form:
\[ \tilde{Q}_d(\epsilon) = \prod_{\substack{d \leq j \leq D \\ type(j-1) = det}} q_j\left(\frac{\epsilon}{(2 \sqrt{2} deg L \kappa)^{\frac{D-j}{2}}} \right) \cdot K^c deg^c L^c \kappa^c \cdot \prod_{\substack{d \leq j \leq D\\ type(j-1) = non-det}} (B_{j} + 1)^4. \]
The product over the $B_{j}+1$ results in a polynomial expression in the variables $\log(\frac{(2 \sqrt{2} deg L \kappa)^{j}}{y})$ where $0 \leq j \leq \frac{D-d}{2}$. Note that, for sufficiently small $\epsilon$, $\log(deg) \leq B_{d+1}$. Thus, we may bound every expression in $polylog(deg, \frac{1}{\epsilon}$ by $B_{d+1}$. Producing a fairly crude, non-optimal bound, we bound the two products by
\begin{equation*}
    \begin{split}
        \prod_{\substack{d \leq j \leq D \\ type(j-1) = det}} q_j\left(\frac{\epsilon}{(2 \sqrt{2} deg L \kappa)^{\frac{D-j}{2}}} \right) & \prod_{\substack{d \leq j \leq D\\ type(j-1) = non-det}} (B_{j} + 1)^4\\
        & \leq deg^{\alpha m} \left( B_{d+1}+1 \right)^{4c + \beta m}\\
        & \leq deg^{\alpha m} \left( \log\left( \frac{(2 \sqrt{2} deg L \kappa)^{\frac{D-d}{2}}}{\epsilon} \right) \right)^{4c + \beta m}\\
        & \leq deg^{\alpha m} \left(\frac{D-d}{2}\right)^{\max(4, \beta)(D-d)} \log\left( \frac{2 \sqrt{2} deg L \kappa}{\epsilon} \right)^{\max(4, \beta)(D-d)}\\
        & = O\left(deg^{\alpha m} (D-d)^{\max(4, \beta)(D-d)} \log\left( \frac{deg L \kappa}{\epsilon} \right)^{\max(4, \beta)(D-d)} \right).
    \end{split}
\end{equation*}

We thus obtain
\begin{equation*}
    \begin{split}
        Q_d(U_{d,\epsilon}) & = O\left( \frac{1}{\epsilon} deg^{\alpha m} K^c deg^c L^c \kappa^c (D-d)^{\max(4, \beta)(D-d)}  \log\left( \frac{deg L \kappa}{\epsilon} \right)^{\max(4, \beta)(D-d)}  \right)\\
        & = O\left( \frac{1}{\epsilon} deg^{\alpha \frac{D-d}{2}} K^{\frac{D-d}{2}} deg^{\frac{D-d}{2}} L^{\frac{D-d}{2}} \kappa^{\frac{D-d}{2}} (D-d)^{\max(4, \beta)(D-d)} \log\left( \frac{deg L \kappa}{\epsilon}\right)^{\max(4, \beta)(D-d)} \right).
    \end{split}
\end{equation*}
If $\alpha = \frac{1}{2}$, this turns into
\[ Q_d(U_{d,\epsilon}) = O\left( \frac{1}{\epsilon} deg^{\frac{D-d}{4}} K^{\frac{D-d}{2}} deg^{\frac{D-d}{2}} L^{\frac{D-d}{2}} \kappa^{\frac{D-d}{2}} (D-d)^{\max(4, \beta)\frac{D-d}{2}} \log\left( \frac{deg L \kappa}{\epsilon}\right)^{\max(4, \beta)\frac{D-d}{2}} \right). \]
This finishes the induction of the complexity estimate and the proof
\end{proof}

\begin{remark} \label{rem:improving.the.superexponential.term}
    The $D^D$-term is of course not desirable, since it creates a situation where an advantage over classical algorithms can only achieved if $deg \gg D$. When following through the induction, we can see that the $D^D$-term arises because the input into the recursive call of $BuildUnitaryEstimator$ has to lower the error-bound. If this error-reduction could be reset either during the deterministic, or non-deterministic step, the $D^D$-term disappears.

    This can be approached in at least three ways while staying with the overall structure of our algorithm: First, once can try to improve the algorithms for the passing between layers. Second, one can try to improve the structure of the algorithm so that the stricter error bounds become unnecessary. Third, the factor accrued by $\epsilon$ disappears if $2 \sqrt{2} deg L \kappa \leq 1$. This might be achievable for particular examples: note how $deg L = 1$ in the expectiminimax case, but it is highly unclear how to deal with the rest of the constants.
\end{remark}

\section{Proof of Lemma \ref{lem:coherent.composition}} \label{app:coherent.composition}

\begin{proof}[Proof of Lemma \ref{lem:coherent.composition}]
    \grn{The displayed form of $\mathcal{C}\ket{v}\ket{0}\ket{0}$ is no restriction. Any state on the $\omega$-register tensored with the rest can be written that way: group the amplitude by the content of the $\omega$-register, let $p^{(v)}_\omega$ be the squared norm of that block, and absorb its phase into $\ket{g^{(v)}_\omega}$. So $\mathcal{C}$ is a synthesizer for $p^{(v)}$ in the sense of \cite[Def.~2.2 and Rem.~2.6]{kothariMeanEstimationSourceCode2023}, whose garbage vectors are explicitly permitted, and by \cite[Rem.~2.12]{kothariMeanEstimationSourceCode2023} we may take the value map to be the identity. Condition (i) makes $\mathcal{C}$ admissible as such a circuit and supplies the controlled versions their algorithm uses.}

    \grn{Their algorithm is built from controlled-$\mathcal{C}$, controlled-$\mathcal{C}^\dagger$, the reflection $2\ket{0}\!\bra{0} - \mathrm{Id}$ and a phase oracle acting on the $\omega$-register alone. It never inspects $\ket{g^{(v)}_\omega}$, and its only hypothesis on the random variable is the second-moment bound (iii). Corollary 3.2 of \cite{sunOptimalQuantumSpeedups2026} adds a clipping of the output into $[-s,s]$ and a median over repetitions, both reversible arithmetic on the output register that may stay uncomputed. This gives the stated guarantee and cost for each fixed $v$.}

    \grn{By (ii) the result is one fixed circuit $W$, so $W(\ket{v}\otimes\ket{0}\ket{0}) = \ket{v} \otimes (\text{output for } v)$ on every basis state $\ket{v}$, and by linearity on every superposition of them. This property lets the layer above call $U$ with its vertex register entangled.}

    \grn{Finally, an approximate or failing subroutine inside $\mathcal{C}$ is not an exception to any of the above. Its failure branch is a branch of $\mathcal{C}\ket{v}\ket{0}\ket{0}$ like any other, so it is already part of $p^{(v)}$; it changes $\mathbb{E}[X_v]$ and $\mathbb{E}[X_v^2]$ and nothing else. The telescoping in (\ref{eq:unitary.RMSE.bound}) and the RMSE induction of Theorem \ref{thm:unitary.general} account for the bias this introduces.}
\end{proof}




\grn{Thus the composition needs no separate hypothesis. The properties of \cite[Cor.~3.2]{sunOptimalQuantumSpeedups2026} that we use are black-box in the sample preparer, and nothing in them distinguishes an inner random variable produced by a search from one produced by an evaluation. The precedent in the literature is stronger than it first appears: \cite[Algorithm 6]{sunOptimalQuantumSpeedups2026} already nests approximate, occasionally failing estimators for $D$ layers, and \cite[Thm.~3.2]{blanchetNonlinearQuantumMonteCarlo2025} does so at two.}

\grn{Condition (ii) also locates the obstruction to D\"urr--H\o yer better than the discussion before Algorithm \ref{alg:min-max-binary-search} does. Intermediate measurements are not by themselves the problem; \textcite{kothariMeanEstimationSourceCode2023} observe that deferring them brings a circuit that contains them into synthesizer form. What (ii) forbids is a subroutine whose cost depends on the branch, and D\"urr--H\o yer is variable-time: it is adaptive and its query count is a random variable. The same obstruction forces the level schedule of \cite{sunOptimalQuantumSpeedups2026} to be derandomised. It also forces the extremum routine of Algorithm \ref{alg:min-max-binary-search} to run a fixed number of rounds.}

\section{Proof of Lemma \ref{lem:search.step}} \label{app:proof.search}

\begin{proof}[Proof of Lemma \ref{lem:search.step}]
    \grn{For each $j$ let $p_j$ be the probability that one read of child $j$, followed by the comparison with $\grn{z}$, yields the bit $1$. We assume nothing about $p_j$ when $|\gamma_j - \grn{z}| \leq \epsilon$; for the other children the hypothesis gives $p_j \leq \eta$ if $\gamma_j \leq \grn{z} - \epsilon$ and $p_j \geq 1-\eta$ if $\gamma_j > \grn{z} + \epsilon$. Let $\mathcal{A}$ prepare the uniform superposition over $j$, apply the read and write the comparison bit, and let $P$ project onto comparison bit $1$. Then $\lVert P\mathcal{A}\ket{0}\rVert^2 = a^2:= \frac{1}{deg}\sum_j p_j$, and $Q:= -\mathcal{A}S_0\mathcal{A}^\dagger S_P$\vio{, where $S_0 := \mathrm{Id} - 2\ket{0}\!\bra{0}$ reflects about the initial state of all registers and $S_P := \mathrm{Id} - 2P$ reflects about the marked subspace, and $Q$ is the Grover iteration as in \cite{brassardQuantumAmplitudeAmplification2002}),} acts on the span of $P\mathcal{A}\ket{0}$ and $(1-P)\mathcal{A}\ket{0}$ as a rotation by $2\theta$ with $\sin\theta = a$, whatever garbage the read carries \cite{brassardQuantumAmplitudeAmplification2002}; after $k$ iterations the amplitude on $P$ is $\sin((2k+1)\theta)$. Since $a$ is unknown, one \emph{set} consists of the runs $Q^k\mathcal{A}\ket{0}$ for $k = 0, 1, 2, 4, \dots, 2^{\lceil \log_2\sqrt{deg}\rceil}$ on separate registers, its flag being the OR of their comparison bits; one \emph{repetition} is the OR of four independent sets; and \emph{found} is the majority of $R = 18\lceil\log(2/\delta')\rceil$ repetitions. All of this is reversible arithmetic on flag bits, and a repetition costs $O(\sqrt{deg})$ applications of the read. The fixed-point search of \cite{yoderFixedPointQuantumSearch2014} could replace the doubling schedule. Three configurations of the true values are possible.

    If $\max_j \gamma_j > \grn{z} + \epsilon$, then $a^2 \geq (1-\eta)/deg$, so $\theta \geq \sin\theta \geq \frac{1}{2\sqrt{deg}}$. If $\theta \geq \frac{\pi}{6}$ the run $k=0$ fires with probability $a^2 \geq \frac14$. Otherwise $(2k+1)\theta$ grows by a factor below $2$ along the schedule after the first step, from $3\theta < \frac{\pi}{2}$ to at least $2\sqrt{deg}\,\theta \geq 1 > \frac{\pi}{4}$, so some run has $(2k+1)\theta \in [\frac{\pi}{4},\frac{\pi}{2}]$ and fires with probability at least $\frac12$. So a set fires with probability at least $\frac14$, a repetition with probability at least $1-(\frac34)^4 > \frac23$, and by Hoeffding's inequality the majority reports \emph{not found} with probability at most $e^{-R/18} \leq \frac{\delta'}{2}$. In this configuration \emph{found} $\Rightarrow$ $\max_j\gamma_j > \grn{z}-\epsilon$ holds trivially.

    If $\max_j \gamma_j \leq \grn{z} - \epsilon$, then $a^2 \leq \eta$ and the run with $k$ iterations fires with probability $\sin^2((2k+1)\theta) \leq (2k+1)^2 \grn{\theta^2 \leq (2k+1)^2 \tfrac{\pi^2}{4}} a^2$\grn{, using $\theta \leq \frac{\pi}{2}\sin\theta$ on $[0,\frac{\pi}{2}]$}. So a set fires with probability at most $\sum_k (2k+1)^2\eta \leq \grn{\tfrac{\pi^2}{4}\big(22\,deg + 16\sqrt{deg} + \log_2 deg + 2\big)\,\eta \leq 100}\,deg\,\eta$ \grn{for $deg \geq 2$} and a repetition with probability at most $\grn{400}\,deg\,\eta \leq \frac14$ for $\grn{C_0 = 1600}$; by Hoeffding the majority reports \emph{found} with probability at most $e^{-R/8} \leq \frac{\delta'}{2}$. In this configuration \emph{not found} $\Rightarrow$ $\max_j\gamma_j \leq \grn{z}+\epsilon$ holds trivially.

    If $\grn{z} - \epsilon < \max_j \gamma_j \leq \grn{z} + \epsilon$, both implications hold whatever the verdict. So the two implications fail with probability at most $\delta'$, and we have assumed nothing about the comparison bits of the children within $\epsilon$ of the pivot. For $\min$ the roles of the two inequalities are exchanged.}
\end{proof}

\section{Proof of Proposition \ref{prop:composed.lower.bound}} \label{app:proof.lower}

\begin{proof}[Proof of Proposition \ref{prop:composed.lower.bound}]
    \grn{Write $m_2 := m_2(\epsilon)$, so that $n := deg^{m_2}$ satisfies $\frac{1}{4\epsilon} \leq n \leq \frac{deg}{4\epsilon}$, and put $m_1 := m - m_2 \geq 1$. We restrict the inputs of $T$ as follows. In the top $2m_1$ layers every chance node receives $deg$ identical subtrees, so these chance layers act as the identity and the top part computes the alternating $\max$--$\min$ formula $f$ of depth $m_1$ and fan-in $deg$ on the values of the $deg^{m_1}$ distinct subtrees rooted at depth $2m_1$, which we call blocks. Inside every block, every $\max$ and $\min$ node receives $deg$ identical subtrees, so the value of the block is the mean of $n = deg^{m_2}$ distinct boolean variables, each duplicated $deg^{m_1}$ times. We impose the promise that this mean lies in $[0, \frac12 - 2\epsilon] \cup [\frac12 + 2\epsilon, 1]$, which is realisable since $4\epsilon n \geq 1$. Let $g$ be the partial function on $\{0,1\}^n$ that outputs $1$ if the mean is at least $\frac12 + 2\epsilon$ and $0$ if it is at most $\frac12 - 2\epsilon$. Since $\max$ and $\min$ commute with thresholding at $\frac12$, the value of $T$ on such an input is at least $\frac12 + 2\epsilon$ if $(f \circ g^{deg^{m_1}})(x) = 1$, where we read $f$ as the AND--OR formula with OR at the $\max$ nodes, and at most $\frac12 - 2\epsilon$ if it is $0$. An estimate with RMSE at most $\epsilon$ lies on the correct side of $\frac12$ except with probability at most $\frac14$ by Chebyshev's inequality, and a query to a leaf of $T$ reveals one variable of $f \circ g^{deg^{m_1}}$. So $Q_{\epsilon}(T) \geq Q(f \circ g^{deg^{m_1}})$, where the right-hand side is the bounded-error query complexity.

    By \cite[Thm.~1.1]{leeQuantumQueryComplexityState2011}, $Q(h) = \Theta(\mathrm{Adv}^{\pm}(h))$ for every function $h$ with finite domain and range, and by \cite[Lemma~5.2]{leeQuantumQueryComplexityState2011}, $\mathrm{Adv}^{\pm}(f \circ g^{deg^{m_1}}) \geq \mathrm{Adv}^{\pm}(f)\,\mathrm{Adv}^{\pm}(g)$, since $f$ is total and $g$ has range $\{0,1\}$. For the outer function, $\mathrm{Adv}^{\pm}(f) \geq \mathrm{Adv}(f) = deg^{m_1/2}$: the positive-weight adversary bound of $\mathrm{AND}_{deg}$ and of $\mathrm{OR}_{deg}$ is $\sqrt{deg}$ and it is multiplicative under composition \cite[Thm.~11]{hoyerNegativeWeightsMake2007}. This is the $\Omega(\sqrt N)$ bound of \cite{barnumLowerBoundQuantum2004} for read-once formulae on $deg^{m_1}$ variables. For the inner function, distinguishing Hamming weight at most $\frac n2 - 2\epsilon n$ from at least $\frac n2 + 2\epsilon n$ on $n$ bits requires $\Omega(\sqrt{t(n-t)}/\Delta)$ queries with $t \approx \frac n2$ and $\Delta = 4\epsilon n \geq 1$ \cite{nayakQuantumQueryComplexity1999}, so $Q(g) = \Omega(\epsilon^{-1})$ and hence $\mathrm{Adv}^{\pm}(g) = \Omega(\epsilon^{-1})$. Altogether $Q_{\epsilon}(T) \geq \Omega( deg^{m_1/2} \epsilon^{-1}) = \Omega(deg^{(m-m_2)/2}\epsilon^{-1})$. For $\frac{1}{4\,deg} \leq \epsilon \leq \frac18$ we have $m_2 = 1$, and in general $deg^{m_2} \leq \frac{deg}{4\epsilon}$ gives $deg^{-m_2/2} \geq (4\epsilon/deg)^{1/2}$ and hence the bound $\Omega(deg^{\frac{m-1}{2}}\epsilon^{-1/2})$.}
\end{proof}

\end{document}